\documentclass[preprint,12pt,authoryear]{elsarticle}

\usepackage{amsmath}
\usepackage{amssymb}
\usepackage{amsfonts}
\usepackage{amsthm}
\usepackage{mathtools}
\usepackage{enumerate}
\usepackage{comment}
\usepackage{braket}

\usepackage[utf8]{inputenc}
\usepackage[T1]{fontenc}

\usepackage{dsfont}

\usepackage{mathrsfs}

\usepackage[%
cal=cm,
]
{mathalfa}

\usepackage{accents}

\usepackage[dvipsnames,svgnames]{xcolor}
\colorlet{MyBlue}{DodgerBlue!60!Black}
\colorlet{MyGreen}{DarkGreen!85!Black}

\usepackage{fullpage}

\usepackage{subfigure}
\usepackage{tikz}
\usetikzlibrary{calc,patterns}
\usepackage{changepage}

\usepackage{acronym}
\usepackage{latexsym}
\usepackage{longtable}
\usepackage{wasysym}
\usepackage{xspace}
\usepackage[normalem]{ulem}

\usepackage{hyperref}
\hypersetup{
colorlinks=true,
linktocpage=true,
pdfstartview=FitH,
breaklinks=true,
pdfpagemode=UseNone,
pageanchor=true,
pdfpagemode=UseOutlines,
plainpages=false,
bookmarksnumbered,
bookmarksopen=false,
bookmarksopenlevel=1,
hypertexnames=true,
pdfhighlight=/O,
urlcolor=MyBlue!60!black,linkcolor=MyBlue!70!black,citecolor=DarkGreen!70!black, 
pdftitle={},
pdfauthor={},
pdfsubject={},
pdfkeywords={},
pdfcreator={pdfLaTeX},
pdfproducer={LaTeX with hyperref}
}

\numberwithin{equation}{section}  
\usepackage[sort&compress,capitalize,nameinlink]{cleveref}
\crefname{app}{Appendix}{Appendices}

\crefrangeformat{equation}{\upshape(#3#1#4)\textendash(#5#2#6)}

\usepackage[textwidth=30mm]{todonotes}

\newcommand{\debug}[1]{#1}

\theoremstyle{plain}
\newtheorem{theorem}{Theorem}

\newtheorem*{corollary*}{Corollary}
\newtheorem{lemma}[theorem]{Lemma}
\newtheorem{proposition}[theorem]{Proposition}

\theoremstyle{definition}
\newtheorem{definition}[theorem]{Definition}
\newtheorem*{definition*}{Definition}

\newtheorem*{hypothesis*}{Hypothesis}

\theoremstyle{remark}
\newtheorem{remark}[theorem]{Remark}
\newtheorem*{remark*}{Remark}
\newtheorem*{notation*}{Notational remark}
\newtheorem{example}[theorem]{Example}

\newcommand{\canb}{\boldsymbol{\debug e}}

\newcommand{\canbres}{\debug \eta}
\newcommand{\canbresprof}{\boldsymbol{\canbres}}

\DeclareMathOperator{\Prob}{\mathsf{\debug{P}}}
\DeclareMathOperator{\Expect}{\mathsf{\debug{E}}}

\newcommand{\cost}{\debug c}
\newcommand{\costprof}{\boldsymbol{\cost}}
\newcommand{\Cost}{\debug C}
\newcommand{\costclass}{\mathcal{\debug \Cost}}

\newcommand{\eqcost}{\debug \lambda}
\newcommand{\eqcostprof}{\boldsymbol{\eqcost}}

\newcommand{\game}{\mathcal{\debug G}}

\newcommand{\eq}[1]{\widehat#1}

\newcommand{\activ}[1]{\widehat#1}

\DeclareMathOperator{\SC}{\mathsf{\debug{SC}}}

\newcommand{\graph}{{\debug G}}
\newcommand{\vertices}{\mathcal{\debug V}}
\newcommand{\edges}{\mathcal{\debug E}}

\newcommand{\vertex}{\debug v}

\newcommand{\edge}{\debug e}

\newcommand{\source}{\mathsf{\debug O}}

\newcommand{\sink}{\mathsf{\debug D}}

\newcommand{\rate}{\debug \mu}
\newcommand{\ratealt}{\debug \nu}
\newcommand{\rateprof}{\boldsymbol{\rate}}
\newcommand{\rateprofalt}{\rateprof'}

\newcommand{\flow}{\debug f}
\newcommand{\flows}{\mathcal{\debug F}}
\newcommand{\flowprof}{\boldsymbol{\flow}}

\newcommand{\load}{\debug x}

\newcommand{\loads}{\mathcal{\debug X}}
\newcommand{\loadprof}{\boldsymbol{\load}}

\newcommand{\nRoutes}{\debug P}
\newcommand{\routes}{\mathcal{\debug \nRoutes}}
\newcommand{\route}{\debug p}

\newcommand{\argdot}{\,\cdot\,}

\newcommand{\diff}{\ \textup{d}}

\newcommand{\ie}{i.e.,\ }
\newcommand{\eg}{e.g.,\ }

\newcommand{\zerovec}{\boldsymbol{\debug 0}}

\newcommand{\perturb}{\debug \varphi}

\newcommand{\regime}{\debug \varrho}
\newcommand{\regimeprof}{\boldsymbol{\regime}}

\newcommand{\valueW}{\debug V}
\newcommand{\vinf}{\debug v}

\newcommand{\eqcostedge}{\debug \tau}

\newcommand{\eqcostedgeprof}{\boldsymbol{\eqcostedge}}

\newcommand{\irun}{\debug i}

\newcommand{\zvar}{\debug z}
\newcommand{\zvarprof}{\boldsymbol{\zvar}}

\newcommand{\Rpos}{\mathbb R_+}

\newcommand{\region}{\debug \Gamma}

\newcommand{\var}{\debug t}

\newcommand{\commodities}{\debug{\mathcal H}}
\newcommand{\commodity}{\debug h}
\newcommand{\commodityalt}{\commodity'}
\newcommand{\resources}{\debug{\mathcal R}}
\newcommand{\resource}{\debug r}
\newcommand{\resourcealt}{\debug {\resource'}}

\newcommand{\strategies}{\debug{\mathcal S}}
\newcommand{\strategy}{\debug s}
\newcommand{\strategyalt}{\debug {\strategy'}}
\newcommand{\strategiesprof}{\boldsymbol{\strategies}}
\newcommand{\routesprof}{\boldsymbol{\routes}}
\newcommand{\resourcesprof}{\boldsymbol{\resources}}

\newcommand{\prim}{\mathsf{\debug P}}
\newcommand{\dual}{\mathsf{\debug D}}
\newcommand{\solset}{\mathsf{\debug S}}

\newcommand{\comonf}{\debug \psi}

\DeclareMathOperator{\co}{\debug{co}}

\newcommand{\eqivgame}[1]{\breve#1}

\newcommand{\R}{\mathbb{R}}

\newcommand{\reals}{\mathbb{R}}

\DeclarePairedDelimiter{\braces}{\{}{\}}

\DeclarePairedDelimiter{\parens}{(}{)}

\DeclarePairedDelimiter{\abs}{\lvert}{\rvert}

\DeclarePairedDelimiterX{\inner}[2]{\langle}{\rangle}{#1,#2}
\DeclarePairedDelimiterX{\setdef}[2]{\{}{\}}{#1:#2}

\DeclarePairedDelimiterXPP{\probof}[1]{\Prob}{(}{)}{}{%

#1}

\DeclarePairedDelimiterXPP{\exof}[1]{\Expect}{[}{]}{}{%

#1}

\newacro{ACG}{atomic congestion game}
\newacro{ACGSD}{atomic congestion game with stochastic demand}
\newacro{CRG}{constrained routing game}
\newacro{PoA}{price of anarchy}
\newacro{PoS}{price of stability}
\newacro{SC}{social cost}
\newacro{SEC}{social expected cost}
\newacro{SO}{social optimum}
\newacro{SOC}{socially optimum cost}

\newacro{PNMC}{parallel-network with multiple-commodities}
\newacro{MES}{monotonic equilibrium selection}

\newacro{NE}{Nash equilibrium}
\newacroplural{NE}[NE]{Nash equilibria}
\newacro{BNE}{Bayesian Nash equilibrium}
\newacroplural{BNE}[BNE]{Bayesian Nash equilibria}
\newacro{PNE}{pure Nash equilibrium}
\newacroplural{PNE}[PNE]{pure Nash equilibria}

\newacro{WE}{Wardrop equilibrium}
\newacroplural{WE}[WE]{Wardrop equilibria}

\newacro{KKT}{Karush\textendash Kuhn\textendash Tucker}
\newacro{OD}[OD]{origin-destination}
\newacro{BPR}{Bureau of Public Roads}

\newacro{SP}{series-parallel}
\newacro{CSP}{constrained series-parallel}

\journal{arXiv}

\begin{document}

\begin{frontmatter}

\title{Monotonicity of Equilibria in Nonatomic  Congestion Games}

\author[labelRoberto]{Roberto Cominetti}

\affiliation[labelRoberto]{organization={Facultad de Ingenier{\'\i}a y Ciencias, Universidad Adolfo Ib\'a\~nez},
            addressline={Diagonal las Torres 2640}, 
            city={Pe{\~n}alol{\'e}n},
            postcode={7910000}, 
            state={Región Metropolitana},
            country={Chile}}
            
\author[labelValerio]{Valerio Dose}

\affiliation[labelValerio]{organization={Dipartimento di Ingegneria Informatica, Automatica e Gestionale, ``Sapienza'' Universit\`a  di Roma},
            addressline={Via Ariosto 25}, 
            city={Roma},
            postcode={00185}, 
            country={Italy}}

\author[labelMarco]{Marco Scarsini}

\affiliation[labelMarco]{organization={Dipartimento di Economia e Finanza, Luiss University},
            addressline={Viale Romania 32}, 
            city={Roma},
            postcode={00197}, 
            country={Italy}}

\begin{abstract}
This paper studies the monotonicity of  equilibrium costs and equilibrium  loads in nonatomic congestion games, in response to variations of the demands. 
The main goal is to identify conditions under which a paradoxical non-monotone behavior can be excluded. 
In contrast with routing games with a single commodity, where the network topology is the sole determinant factor for monotonicity, for general congestion games with multiple commodities the structure of the strategy sets plays a crucial role.

We frame our study in the general setting of congestion games, with a special focus on
singleton congestion games, for which we establish the monotonicity of equilibrium loads with respect to every demand. 
We then provide conditions for comonotonicity of the equilibrium loads, \ie we investigate when they jointly increase or decrease after variations of the demands. 
We finally extend our study from singleton congestion games to the larger class of \acl{CSP} congestion games, whose structure is reminiscent of the concept of a \acl{SP} network.
\end{abstract}


\end{frontmatter}


\section{Introduction}

Decision making in a multi-agent strategic context is prone to various paradoxes that are impossible in a single-agent framework.
For instance, expanding the feasible choice set produces a 
better outcome in single-agent optimization, but, in a game, it may give rise to an equilibrium that is worse for all players. 
Analogously, more information is beneficial in single-agent decision making under risk, but may induce worse Bayes-Nash equilibria in a game. 

Several paradoxes arise in routing games. 
These games represent situations where roads to go from one origin to the corresponding destination are chosen strategically by travellers in a way that minimizes their traveling time. 
The nonatomic version of these games is a good approximation of situations with a large number of travelers. 
In nonatomic games the standard equilibrium concept, due to \citet{War:PICE1952}, prescribes that, for each \ac{OD} pair, only the paths with the smallest traveling time are used and they all have the same traveling time. 
A famous paradox in  routing games, due to \citep{Bra:U1968,BraNagWak:TS2005}, shows that adding an edge to a network can make the traveling time worse for all players.
Other paradoxes arise in this class of games. 
For instance, although one could expect that an increase in traffic demand would make the traveling time higher across the network, this is not always the case. In fact, while \citet{Hal:TS1978} proved that---\emph{ceteris paribus}---an increase in the demand of one \ac{OD} pair increases the traveling time of this \ac{OD},  \citet{Fis:TRB1979} showed that an increase of traffic demand of one \ac{OD} pair can be beneficial for some other \ac{OD} pair by decreasing its traveling time.
Even in networks with a single \ac{OD} pair, an increment in the traffic demand may decrease the equilibrium load on some edges in the network. 
These paradoxes will be examined in detail in \cref{ex:Fisk,ex:Wheatstone}.

Networks in which the equilibrium loads of all the edges increase with the travel demand of every \ac{OD} pair are more predictable and easier to handle for a social planner, because an edge is never used below a certain level of demand and is always used above that level.
The goal of this paper is precisely to understand when the equilibrium travel times and edge loads are monotone in the  demand, so that the paradoxical phenomena observed in the above  examples cannot happen. 
Rather than focusing on routing games, we will state our results for the wider class of congestion games, of which routing games are a
significant but particular example.

\subsection{Our Results}
\label{susc:our-results}

Nonatomic congestion games are defined by a finite set of resources and a finite set of commodities.
Each commodity has a demand that can be satisfied by different strategies in a strategy set, where each strategy is a subset of the resource set.
In a Wardrop equilibrium each resource has a nonnegative load (a fraction of the total  demand), which varies with the demand vector. 

The first part of our paper (\cref{sc:singleton}) focuses on singleton congestion games, in which every strategy contains only one resource. 
We start by proving an equilibrium selection result for this class of games: \cref{thm:singleton-congestion-games} shows that, even when there exist multiple equilibrium flows, one can always select one equilibrium whose corresponding resource loads are monotone increasing with respect to each demand.

We then use the notion of comonotonicity, which captures the idea that different resource loads jointly increase or decrease upon variations of the demands.
\cref{thm:parallel-regions} provides some structural results about the demand regions where different subsets of resources are used in equilibrium and how these resources become active or inactive as the demands vary.
This analysis allows us to identify regions of the space of demands where the equilibrium loads are comonotonic.

The following section is devoted to games that are more general than singleton congestion games.
\cref{prop:congestion-to-routing} shows that every congestion game can be suitably represented as a routing game that is subject to some restrictions, \ie not every path from an origin to a destination is feasible. 
Then \cref{pr:constrained-series-parallel} extends the monotonicity properties of \cref{sc:singleton} to a class of games that is obtained from singleton congestion games by  applying the series and parallel operations.
Finally \cref{pr:constrained-series-parallel-routing-description} relates \acl{CSP} games to routing games.
These results shed light on the features that produce the non-monotonicity paradoxes, and highlights the difference between the single- and multiple-\ac{OD} networks: for  routing games with a single \ac{OD} pair, the network topology is the sole relevant factor that guarantees the monotonicity of equilibrium loads, whereas  for multiple \acp{OD} the structure of the set of feasible routes  plays a crucial role.

\subsection{Related Work}
\label{susc:related-work}

Several authors studied the sensitivity of Wardrop equilibria in  routing games with respect to changes in the demand. 
\citet{Hal:TS1978} observed that, when the costs are strictly increasing, the equilibrium loads depend continuously on the demands. 
\citet{Pat:TS2004} and \citet{JosPat:TRB2007} studied the directional differentiability (or lack thereof) of equilibrium costs and loads, whereas  \citet{ComDosSca:arXiv2023} studied differentiability along a curve in the space of demands. 
Specific cases of differentiability, were also considered in \citet{Pra:thesis2014}. 

As mentioned previously, \citet{Hal:TS1978} proved that the equilibrium cost of an \ac{OD} pair increases when the demand of that \ac{OD} pair grows.
Some positive results concerning the monotonicity of equilibrium loads in \acl{SP} single-commodity networks can be found in \citet{KliWar:MOR2022} for piece-wise linear costs and in \citet{ComDosSca:MP2021} for general nondecreasing costs.

Traffic equilibria in routing games exhibit a multitude of paradoxes. 
The most famous, due to 
\citet{Bra:U1968},
shows that removing an edge from a network could actually improve the equilibrium cost for all players
(see \cref{fig:classic_braess}). 
Also surprising is the fact observed by \citet{Fis:TRB1979} that an \ac{OD} can reduce its cost and benefit from an increase in the demand of a different \ac{OD}, even after doubling all the demands.
\citet{FisPal:TRA1981} showed that such paradoxical phenomena could be observed in real life in the City of Winnipeg, Manitoba, Canada.
\citet{DafNag:TRB1984} studied how equilibrium costs are affected by changes in the travel demand or addition of new routes under a more general 
non-separable cost structure. 
A related paradoxical phenomenon was studied by \citet{MehHor:IEEETCNS2020} in a model with both regular and autonomous vehicles: 
despite the fact that autonomous vehicles are more efficient by allowing shorter headways and distances, 
replacing regular with autonomous vehicles may increase the total network delay.

A particularly simple class of congestion games is the one of singleton congestion games where each strategy comprises a single resource. Different variants of these type of games have been considered in the literature,
including atomic weighted and unweighted players, with splittable or unsplittable loads, as well as nonatomic games. 

For \emph{atomic splittable} singleton games, 
\citet{HarTim:LNCS2017}
developed a polynomial time algorithm to compute a Nash equilibrium with player-specific affine costs. 
In a different direction, \citet{BilVin:ESA2017}
investigated how the structure of the players’ strategy sets affects the efficiency in singleton load balancing games. 
Atomic splittable singleton games have also been used to model the charging strategies of a population of electric vehicles \citep{MaCalHis:IEEETCST2013,DeoMarPra:IFACPOL2017,NimMedRatShaSmiHal:IEEETITS2020}.
In a related but different direction, \citet{CasMarGatCon:AI2019} studied the computational complexity of finding Stackelberg equilibria in games where one player acts as leader and the others as followers.

For \emph{atomic unsplittable} singleton  games,  
\citet{GaiSch:WINE2007} provided upper and lower bounds on the price of anarchy, distinguishing between restricted and unrestricted strategy sets, weighted and unweighted players, and linear vs. polynomial costs. 
\citet{FotKonKouMavSpi:TCS2009} studied the combinatorial structure and computational complexity of Nash equilibria,  including the problems of deciding the existence of pure  equilibria, computing pure/mixed equilibria, and computing the social cost of a given mixed equilibrium. 
\citet{GaiLucMavMon:TCS2010} studied weighted atomic unsplittable routing games on a parallel-edge network where each user can only route over a restricted set of edges. 
They developed a polynomial time algorithm for the model where the edge costs are identical and linear, and  both player weights, and edge capacities are integer.    
\citet{HarKli:MOR2012} characterized the classes of cost functions  that guarantee the existence of pure equilibria for weighted  routing games and singleton congestion games.

Finally, in the \emph{nonatomic} setting, which is the focus of our paper, \citet{GonTen:EC2016} used a clever hydraulic system representation to study asymmetric singleton congestion games, presenting applications in the home internet and cellular markets, as well as in cloud computing. 
Another recent application of nonatomic singleton congestion games to hospital choice in healthcare systems is discussed in \citet{vdK:arXiv2023}.
In the special case of routing games, singleton games correspond to parallel networks. Despite its simple topology they are nevertheless of interest in the literature \citep[see, \eg][]{AceOzd:MOR2007,Wan:EJOR2016,HarSchVer:EJOR2019}).
\citet{FujGoeHarPeiZen:MOR2017} considered nonatomic congestion games and used matroid theory to characterize games for which two forms of Braess's paradox cannot occur.
A similar problem was considered by \citet{VerChe:arXiv2023}, who---among other things---study the effect of Braess's paradox at different levels of the demand in single-\ac{OD} routing games with affine costs.

In \cref{sc:parallel-regimes} we use the concept of comonotonicity.
Although its definition is purely analytic and concerns real functions defined on an arbitrary space, the idea originated in various applications in  actuarial science \citep{Bor:E1962}, economic theory \citep{Wil:E1968,Arr:NH1970}, and decision theory \citep{Yaa:E1987,Sch:E1989}. 
A mathematical treatment of the concept---in connection with Choquet capacities---can be found in  \citet{Del:SPUS1970}, who uses the term ``même tableau de variation'' and \citet{Sch:PAMS1986}, who---to the best of our knowledge---was the first to use the term comonotonic in his preprint \citet{Sch:mimeo1984} (there exists a previous version, \citet{Sch:mimeo1982}, which we could not access, so we don't know whether the term was used there or not).
A recent application of comonotonicity to game theory can be found in \citet{KocRujKuh:MP2022}.
The reader is referred to \citet{DhaDenGooKaaVyn:IME2002,PucSca:JMVA2010} for a more  thorough discussion and further  references.

\subsection{Organization of the paper}
\label{susc:organization}

The paper is organized as follows.
\cref{sc:prelim} recalls the standard model of non-atomic congestion games and  reviews the basic properties of equilibria. 
This section includes the definition of \acl{MES} and  comonotonicity. 
\cref{sc:singleton,sc:parallel-regimes} both deal with singleton congestion games. 
\cref{sc:singleton} contains the central monotonicity result, whereas \cref{sc:parallel-regimes} discusses comonotonicity and the structure of the domains associated to different sets of resources.
\cref{sc:beyond-singleton} studies the monotonicity properties of more complex congestion games beyond the case of singleton strategies.
\cref{sc:summary} summarizes the results of our paper and proposes some open problems.
\cref{sc:appendix:proofs} includes some supplementary proofs.
\cref{sc:list-of-symbols} contains a list of the symbols used throughout the paper.

%
%

\section{Congestion Games and Equilibria}
\label{sc:prelim}

In this section we recall the basic concepts and properties of nonatomic congestion games, and we fix the notations used throughout the paper. 
The basic structural elements are:     
\begin{itemize}
\item 
a finite set $\resources$ of \emph{resources} and, for each $\resource\in\resources$, a continuous nondecreasing \emph{cost function} $\cost_{\resource} \colon \reals_{+} \to \reals_{+}$, where $\cost_{\resource}(\load_{\resource})$ represents the cost of resource $\resource$ under a workload $\load_{\resource}$; and

\item 
a finite set $\commodities$ of \emph{commodities} and, for each $\commodity\in\commodities$,  a family $\strategies^{\commodity}\subset2^\resources\setminus\varnothing$ of \emph{feasible strategies}, where every  $\strategy\in\strategies^{\commodity}$ is a nonempty subset  of resources $\strategy\subset\resources$.
\end{itemize}
These elements define a  \emph{congestion game structure} $\game=\parens*{\resources,\costprof,\strategiesprof}$
with $\costprof \coloneqq \parens*{\cost_{\resource}}_{\resource\in\resources}$ the vector of cost functions and 
$\strategiesprof \coloneqq \times_{\commodity\in\commodities}{\strategies^{\commodity}}$ the set of strategy profiles.

Every vector $\rateprof \coloneqq \parens*{\rate^{\commodity}}_{\commodity\in\commodities}$ of \emph{demands} $\rate^{\commodity} \ge 0$, determines a \emph{nonatomic congestion game} $(\game,\rateprof)$ as follows.
For each commodity $\commodity\in\commodities$, a \emph{feasible flow} is a vector $\flowprof^{\commodity} \coloneqq \parens*{\flow_{\strategy}^{\commodity}}_{\strategy\in\strategies^{\commodity}}$ satisfying  
\begin{equation}
\label{eq:feasible-flows-per-commodity}    
\rate^\commodity = \sum_{\strategy\in\strategies^\commodity} \flow_{\strategy}^{\commodity}, \quad \flow_{\strategy}^{\commodity} \ge 0, \text{ for all } \strategy\in\strategies^{\commodity}.
\end{equation}
A family $\flowprof \coloneqq \big(\flowprof^{\commodity}\big)_{\commodity\in\commodities}$, where each $\flowprof^{\commodity}$ is a  feasible flow satisfying \eqref{eq:feasible-flows-per-commodity},  induces aggregate \emph{loads} $\loadprof=(\load_{\resource})_{\resource\in\resources}$ over the resources, given by 
\begin{equation}
\label{eq:edge-loads}
\load_{\resource} \coloneqq \sum_{\commodity\in\commodities}\sum_{\strategy\in\strategies^{\commodity}}\flow_{\strategy}^{\commodity}\mathds{1}_{\{\resource\in\strategy\}},\quad \forall \resource\in\resources,
\end{equation}
which in turn induce \emph{strategy costs},  defined as
\begin{equation}
\label{eq:strategy_cost}
\cost_{\strategy}(\loadprof)\coloneqq\sum_{\resource\in\strategy}\cost_{\resource}(\load_{\resource}), \quad \forall\strategy\subset\resources.
\end{equation}
We call $\flows_{\rateprof}$  the set of \emph{feasible pairs} $(\flowprof,\loadprof)$ 
satisfying \eqref{eq:feasible-flows-per-commodity}
and \eqref{eq:edge-loads}.
We also write $\loads_{\rateprof}$ for the projection of $\flows_{\rateprof}$ on the $\loadprof$ variables, that is, the set of load profiles $\loadprof$ induced by all feasible flow vectors $\flowprof$.

The concept of Wardrop equilibrium is based on the assumption that for each commodity only the strategies with the smallest possible cost are actually used.
A feasible pair $(\flowprof,\loadprof)\in \flows_{\rate}$ is a \emph{Wardrop equilibrium} if there exists a nonnegative vector  $\eqcostprof \coloneqq \parens*{\eqcost^{\commodity}}_{\commodity\in\commodities}$, such that 
\begin{equation}
\label{eq:Wardrop}\forall \commodity\in\commodities,\quad\begin{cases}
\cost_{\strategy}(\loadprof)=\eqcost^{\commodity} & \text{for all }\strategy\in\strategies^{\commodity}\text{ with }\flow_{\strategy}^{\commodity}>0,\\
\cost_{\strategy}(\loadprof)\ge\eqcost^{\commodity} & \text{for all }\strategy\in\strategies^{\commodity}\text{ with }\flow_{\strategy}^{\commodity}=0.
\end{cases}
\end{equation}
The quantity $\eqcost^{\commodity}$ is called the \emph{equilibrium cost} of commodity $\commodity\in\commodities$.
A strategy $\strategy\in\strategies^{\commodity}$ is said to be \emph{active} if $\cost_{\strategy}(\loadprof)=\eqcost^{\commodity}$. Similarly, a resource $\resource\in\resources$ is \emph{active} for commodity $\commodity\in\commodities$ if it belongs to some active strategy.
Clearly, the equilibrium equation implies that every strategy  carrying a strictly positive flow $\flow_{\strategy}^{\commodity}>0$ is necessarily active. 
Note, however, that a strategy with zero flow may still be active as long as its cost matches the minimum.

As shown by \citet{BecMcGWin:Yale1956}, the set of  load profiles induced by  equilibrium flows coincides with the set of optimal solutions of the minimization problem
\begin{equation}
\label{eq:Beckmann}
\min_{\loadprof\in\loads_{\rate}}\sum_{\resource\in\resources}\Cost_\resource(\load_\resource),    
\end{equation}
where $\Cost_\resource(\load_\resource)\coloneqq\int_0^{\load_\resource}\cost_\resource(z)\diff z$. 
Since the cost functions $\cost_\resource$ are continuous and nondecreasing, the above objective function  is convex and differentiable. 
Thus, since $\loads_{\rateprof}$ is a bounded polytope, for every $\rateprof$ there exists at least one optimal solution.

For an equilibrium load profile  $\eq{\loadprof}$, we  define the equilibrium resource costs $\eqcostedge_{\resource} \coloneqq \cost_{\resource}(\eq{\load}_{\resource})$. 
By using Fenchel's duality theory (see \eg \cref{rem:dual} in  \cref{sc:appendix:proofs}, or \citet{Fuk:TRB1984} for the special case of nonatomic routing games), we can prove that the equilibrium resource costs are optimal solutions of the strictly convex dual program 
\begin{equation}\label{eq:Fukushima}
\min_{\eqcostedgeprof}\sum_{\resource\in\resources}\Cost_{\resource}^{*}(\eqcostedge_{\resource}) -
\sum_{\commodity\in\commodities}\left( \rate^{\commodity}\,\min_{\strategy\in\strategies^{\commodity}}\sum_{\resource\in\strategy}\eqcostedge_{\resource}\right),
\end{equation}
where $\Cost_{\resource}^{*}(\argdot)$ is the Fenchel conjugate of $\Cost_\resource(\argdot)$, which is strictly convex. 

Thus, for each $\rateprof$ the equilibrium resource costs $\eqcostedge_\resource$ are uniquely defined  and are the same for all equilibrium loads. 
This implies that the strategy costs 
$\cost_\strategy=\sum_{\resource\in\strategy}\eqcostedge_\resource$ and
equilibrium costs $\eqcost^{\commodity}=\min_{\strategy\in\strategies^{\commodity}}\sum_{\resource\in\strategy}\eqcostedge_\resource$ depend only on $\rateprof$ and not on the particular equilibrium flow under consideration. Thus, also the active strategies and active resources only depend on $\rateprof$.

The \emph{active regime} at demand $\rateprof$ is defined as $\activ{\resourcesprof}(\rateprof)\coloneqq(\activ{\resources}^{\commodity}(\rateprof))_{\commodity\in\commodities}$ with $\activ{\resources}^{\commodity}(\rateprof)$ the set of active resources for commodity $\commodity\in\commodities$. 
We also let  $\rateprof\mapsto\eqcost(\rateprof)$ denote the \emph{equilibrium cost} map, whose basic properties are summarized in the next proposition.

\begin{proposition}
\label{pr:continuity}
Let $\game=(\resources,\costprof,\strategiesprof)$ be a congestion game structure. 
Then  the equilibrium cost map $\rateprof\mapsto\eqcost(\rateprof)$ is continuous and monotone in the sense that $\langle\eqcost (\rateprof_{1})- \eqcost(\rateprof_{2}),\rateprof_{1}-\rateprof_{2}\rangle\ge0$ for every $\rateprof_{1},\rateprof_{2}\in\Rpos^{\commodities}$.
In particular, each component $\eqcost^{\commodity}(\rateprof)$ is nondecreasing with respect to its own demand $\rate^{\commodity}$.
Moreover, the equilibrium resource costs $\eqcostedge_{\resource}(\rateprof)$ are uniquely defined and continuous.
\end{proposition}

\cref{pr:continuity} is a simple extension of \citet[Proposition 3.1]{ComDosSca:MP2021}
to the multi-commodity setting. 
See also \citet{Hal:TS1978} for the case of strictly increasing costs.
For the sake of completeness, we include a proof of \cref{pr:continuity} in \cref{sc:appendix:proofs}.

\begin{remark}
\label{rm:continuity-of-flows-when-costs-invertible}
When the cost functions are strictly increasing, thus invertible, \cref{pr:continuity} implies that the  equilibrium load vector $\loadprof(\rateprof)$  
is unique for every $\rateprof\in\Rpos^{\commodities}$, and the map $\rateprof\mapsto\load(\rateprof)$ is continuous. 
If the costs are just nondecreasing, the equilibrium loads may be non-unique. Here we point out that there exists  some literature about the characterization of games having the so-called uniqueness property \citep[see, \eg][]{Mil:MOR2000,Kon:TS2004,Mil:MOR2005,MeuPra:EJOR2014}. 
A natural question for the case of multiple equilibria is whether there exists a continuous selection $\rateprof\mapsto\load(\rateprof)$. 
\end{remark}

Routing games are an important instance of congestion games. 
In this class of games there is a finite network in the background with a finite set of \ac{OD} pairs (the commodities of the game); edges are the resources and paths from one origin to the corresponding destination are the strategies of the game. 

\citet{Hal:TS1978} proved that in routing games an increase of traffic demand for one \ac{OD} pair---when the remaining demands are kept fixed---weakly  increases the traveling time of this \ac{OD} pair.  
The following example, due to \citet{Fis:TRB1979}, shows that an increase in the traffic demand of one \ac{OD} pair may actually reduce the traveling time of another \ac{OD} pair.
\citet{Fis:TRB1979} showed that it is also possible for the social cost ${\SC}(\rateprof)=\sum_{\commodity\in \commodities} \rate^{\commodity}\eqcost^{\commodity}(\rateprof)$ to decrease along a direction where the total demand $\sum_{\commodity\in\commodities}\rate^{\commodity}$ increases.

\begin{example}
\label{ex:Fisk}
Consider the network depicted in \cref{fig:Fisk} with three \ac{OD} pairs $(a,b)$, $(b,c)$, $(a,c)$.

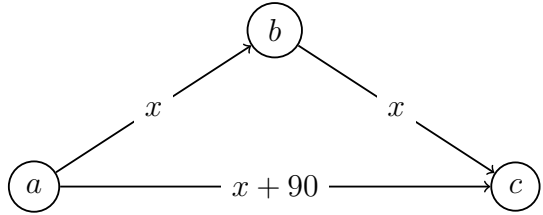
\begin{figure}[ht]
\centering
\hspace{2cm}\begin{tikzpicture}[scale=0.8]
    \node[shape=circle,draw=black,line width=.7pt] (v1) at (0,0)  {$a$}; 
   \node[shape=circle,draw=black,line width=.7pt] (v2) at (4,2.6)  { $b$}; 
   \node[shape=circle,draw=black,line width=.7pt] (v6) at (8,0)  { $c$}; 
   \draw[line width=.7pt,->] (v1) to   node[midway,fill=white] {$\load$} (v2);
   \draw[line width=.7pt,->] (v1) to   node[midway,fill=white] {$\load+90$} (v6);
   \draw[line width=.7pt,->] (v2) to   node[midway,fill=white] {$\load$} (v6);
\end{tikzpicture}
\caption{Fisk's network. }
\label{fig:Fisk}
\end{figure}
\noindent
Let the initial demands be $
\rate^{(a,b)}=1$, $\rate^{(a,c)}=20$,
$\rate^{(b,c)}=100$,
and let the cost functions be as in \cref{fig:Fisk}.
The equilibrium loads are
$\load_{(a,b)}=4$, 
$\load_{(a,c)}=17$, 
$\load_{(b,c)}=103$,  
and the corresponding equilibrium costs are \begin{equation*}
\eqcost^{(a,b)}=4,\quad
\eqcost^{(a,c)}=107,\quad
\eqcost^{(b,c)}=103.    
\end{equation*}  
If we now let the demand $\rate^{(a,b)}$ rise from $1$ to $4$, the new equilibrium loads are
$\load_{(a,b)}=6$,
$\load_{(a,c)}=18$,
$\load_{(b,c)}=102$,   
and the corresponding equilibrium costs are
\begin{equation*}
\eqcost^{(a,b)}=6,\quad
\eqcost^{(a,c)}=108,\quad
\eqcost^{(b,c)}=102.    
\end{equation*}
That is, the increase of $\rate^{(a,b)}$ increases the cost of the edge $ab$ and pushes the $(a,c)$ pair to favor the use of the direct edge $ac$.
This reduces the load on the edge $bc$, which ultimately benefits the pair $(b,c)$  by reducing its cost.

Perhaps more surprising is the fact that this phenomenon may even occur when all the demands increase by the same factor: with demands $\rate^{(a,b)}=60$, $\rate^{(a,c)}=30$, $\rate^{(b,c)}=6$
the equilibrium cost for the third \ac{OD} is $\eqcost^{(b,c)}=24$, and when all the demands are doubled it decreases to $\eqcost^{(b,c)}=18$.

\end{example}

\begin{example}
\label{ex:Wheatstone}
Even in  networks with a single \ac{OD} pair, where the equilibrium cost  increases with the demand, it may happen that the load on some edges decrease after a surge in the demand. 
This can be observed in the classical Wheatstone network depicted in \cref{fig:classic_braess} \citep[see][for the famous paradox that uses this network]{Bra:U1968,BraNagWak:TS2005}. 
\begin{figure}[ht]
\centering
\setcounter{subfigure}{0}
\subfigure[Wheatstone network]
{
\begin{tikzpicture}[thick,scale=0.8, every node/.style={transform shape}]
   \node[shape=circle,draw=black,line width=1pt,minimum size=0.5cm] (v1) at (-3,0)  { $\source$}; 
   \node[shape=circle,draw=black,line width=1pt,minimum size=0.5cm] (v2) at (0,1.5)  {$\vertex_{1}$}; 

   \node[shape=circle,draw=black,line width=1pt,minimum size=0.5cm] (v5) at (0,-1.5)  {$\vertex_{2}$}; 
   \node[shape=circle,draw=black,line width=1pt,minimum size=0.5cm] (v6) at (3,0)  {$\sink$}; 
    
   \draw[line width=1pt,->] (v1) to   node[midway,fill=white] {$\load$} (v2);
   \draw[line width=1pt,->] (v1) to   node[midway,fill=white] {$1$} (v5);
   \draw[line width=1pt,->] (v2) to   node[midway,fill=white] {$0$} (v5);
   
   \draw[line width=1pt,->] (v2) to   node[midway,fill=white] {$1$} (v6);

   \draw[line width=1pt,->] (v5) to   node[midway,fill=white] {$\load$} (v6); 
        
\end{tikzpicture}
\label{fig:graph_Braess}
}
\hspace{.2cm}
\subfigure[Equilibrium flows for different values of the demand $\rate$]{
\raisebox{1.5cm}{
\begin{tabular}{c|c|c|c}
& $\source\,\vertex_{1}\,\vertex_{2}\sink$ & $\source\,\vertex_{1}\sink$ & $\source\,\vertex_{2}\sink$ \\
\hline
$\rate\in[0,1]$ & $\rate$ & $0$ & $0$ \\
\hline
$\rate\in[1,2]$ & $2-\rate$ & $\rate-1$ & $\rate-1$
\\
\hline
$\rate\in[2,+\infty)$ & $0$ & $\rate/ 2$ & $\rate/2$
\end{tabular}}
}
\caption{In the Wheatstone network with three paths and a single \ac{OD} pair, the equilibrium load on the vertical edge $(\vertex_{1},\vertex_{2})$ equals the equilibrium flow on the path $\source\,\vertex_{1}\,\vertex_{2}\sink$ and is decreasing for $\rate\in[1,2]$.}
\label{fig:classic_braess}
\end{figure}
\end{example}

In what follows, we want to determine if a congestion game has an equilibrium selection such that the resource loads are monotone with respect to an increase in any demand. This is made precise in the following definition.

\begin{definition}
A congestion game structure $\game=(\resources,\costprof,\strategiesprof)$ is said to have a \acfi{MES}\acused{MES} 
if there exists an equilibrium load vector $\loadprof(\rateprof)$ 
such that for every resource $\resource\in\resources$ the map $\rateprof\mapsto \loadprof_{\resource}(\rateprof)$ is nondecreasing 
with respect to each component $\rate^{\commodity}$ of the demand vector $\rateprof$.
\end{definition}

In mixed scenarios where some demands increase and other decrease, one may naturally expect that the same holds for the induced equilibrium loads. However, it is still of interest to identify groups of resources whose equilibrium loads vary in the same direction, regardless whether $\rateprof$ and $\rateprofalt$ are  comparable or not. 
In such a case, observing an increase/decrease in the load of a specific resource one can infer that all the remaining loads in the group move in the same direction. 
This property is captured by the  notion of comonotonicity: 
a family of functions $\{\comonf_{i}:\Omega\to\reals\}_{i\in A}$ is \emph{comonotonic} if
for all $i,j\in A$ we have
\begin{equation}
\label{eq:comonotonic}
\forall \omega_{1},\omega_{2}\in\Omega, \quad (\comonf_{i}(\omega_{1})-\comonf_{i}(\omega_{2})) (\comonf_{j}(\omega_{1})-\comonf_{j}(\omega_{2})) \ge 0.
\end{equation}
For singleton congestion games, we will identify subsets of resources whose equilibrium loads exhibit such comonotonic behavior in specific regions of the space of demands $\reals_+^{\commodities}$. 
Informally, we will show that a group of commodities that  share the same equilibrium cost behave as a single commodity, and the loads on the resources used by this group are comonotonic.

%
%

\section{Monotonicity in singleton congestion games}\label{sc:singleton}
In a \emph{singleton congestion game}  each strategy corresponds to a single resource. Thus, for every commodity $\commodity\in\commodities$ the set of feasible strategies $\strategies^\commodity$ can be viewed as a subset $\resources^\commodity\subset\resources$ of the set of resources.
The following result shows that the \ac{MES} property holds in this case.

\begin{theorem}
\label{thm:singleton-congestion-games}
Every singleton congestion game  structure $\game=(\resources,\costprof,\strategiesprof)$ has a \ac{MES}.
\end{theorem}

\begin{proof}
We first prove the result for strictly increasing cost functions, and we then use a regularization argument to address the general case of nondecreasing costs.

Suppose first that 
the costs $\cost_{\resource}(\argdot)$ are strictly increasing. 
We will prove the existence of a \ac{MES} locally by showing that for every demand vector $\rateprof_0\in\Rpos^{\commodities}$
and every commodity $\commodity\in\commodities$, there exists $\varepsilon>0$ such that  
$\load_\resource(\rateprof_0+\var \canb^{\commodity})\ge\load_\resource(\rateprof_0)$ for all $\var\in[0,\varepsilon]$, 
where $\canb^{\commodity}$ is the $\commodity$-th vector of the canonical basis of $\mathbb R^{\commodities}$. 
The global \ac{MES} property throughout the space of demands then follows from the continuity of the map $\rateprof\mapsto\loadprof(\rateprof)$ (see \cref{rm:continuity-of-flows-when-costs-invertible}). 

Let $\resources_0$ be the set of resources such that $\cost_{\resource}(\load_{\resource}(\rateprof_0))=\eqcost^{\commodity}(\rateprof_0)$. 
This set  contains the active resources for commodity $\commodity$ but may also include resources used by other commodities and that are not feasible for $\commodity$.
By continuity of the equilibrium costs (\cref{pr:continuity}), there exists $\varepsilon>0$ such that an increase in the demand for commodity $\commodity$ by an amount $\var$ smaller than  $\varepsilon$ can only affect the equilibrium loads of resources in  $\resources_0$, and therefore for $\resource\notin \resources_0$ and $\var\in[0,\varepsilon]$ we have $\load_\resource(\rateprof_0+\var \canb^{\commodity})=\load_\resource(\rateprof_0)$.
Let us then focus on the resources  $\resource\in\resources_0$. Fix an arbitrary $\var\in[0,\varepsilon]$ and partition $\resources_0$ into the three subsets
\begin{align}
\label{eq:R0+}
\resources_0^{+}&\coloneqq \{\resource\in \resources_0\colon \load_\resource(\rateprof_0+\var \canb^{\commodity})>\load_\resource(\rateprof_0)\},\\
\label{eq:R0-}
\resources_0^{-}&\coloneqq \{\resource\in \resources_0\colon \load_\resource(\rateprof_0+\var \canb^{\commodity})<\load_\resource(\rateprof_0)\},\\
\label{eq:R0=}
\resources_0^{=} &\coloneqq \{\resource\in \resources_0\colon \load_\resource(\rateprof_0+\var \canb^{\commodity})=\load_\resource(\rateprof_0)\}.   
\end{align}
Suppose by contradiction that  $\resources_0^-$ is not empty. 
Since the total demand at $\rateprof_0+\var \canb^{\commodity}$ is strictly larger than the total demand at $\rateprof_0$, whereas  
the total flow on the resources  
$\resources_0^{-} \cup \resources_0^{=}$
decreases, some flow must have been transferred from  $\resources_0^{-} \cup \resources_0^{=}$ to $\resources_0^{+}$.
This implies the existence of a commodity $\commodityalt$ which has feasible resources  in  both $\resources_0^{-} \cup \resources_0^{=}$ and $\resources_0^{+}$, and which sends a positive flow along a resource in $\resources_0^{+}$at demand $\rateprof_0+\var \canb^{\commodity}$.
This contradicts the equilibrium condition for that commodity  because the cost of all resources in $\resources_0^{+}$ is strictly higher than the cost of the resources in $\resources_0^{-} \cup \resources_0^{=}$.
This establishes the existence of a \ac{MES} for the case of strictly increasing costs. 

When  costs $\cost_\resource(\load_\resource)$ are assumed to be just nondecreasing,  we perturb them as $\cost^{\varepsilon}_\resource(\load_\resource) \coloneqq \cost_\resource(\load_\resource)+2\varepsilon \load_\resource$ with $\varepsilon>0$, to make them strictly increasing, and then consider the limit as $\varepsilon$ approaches zero. 
As recalled in \cref{sc:prelim}, the equilibrium flow $\loadprof(\rateprof,\varepsilon)$ for the congestion game structure $\game^{\varepsilon} \coloneqq (\resources,\costprof^{\varepsilon},\strategiesprof)$ is the unique solution of the Beckmann problem \eqref{eq:Beckmann}, which in this case has the form
\begin{equation}
\label{eq:Beckmann-regularized}
\min_{\load\in\loads_{\rate}}\sum_{\resource\in\resources}\Cost_\resource(\load_\resource)+\varepsilon\lVert\loadprof\rVert^2,
\end{equation}
with $\Cost_\resource(\load_\resource) \coloneqq \int_0^{\load_\resource}\cost_\resource(z)\diff z$. 
Tikhonov regularization \citep[see, \eg][section 1.1]{Att:SIAMJO1996}
 tells us that $\loadprof(\rateprof,\varepsilon)$ converges, as $\varepsilon$ approaches zero, to the minimal norm equilibrium $\loadprof_0(\rateprof)$ of the original unperturbed game  $\game$. 
From the previous case of strictly increasing costs, for each $\varepsilon>0$ the map $\rateprof\mapsto\loadprof(\rateprof,\varepsilon)$ is nondecreasing with respect to each demand $\rate^{\commodity}$, and this  property  is inherited by $\rateprof\mapsto\loadprof_0(\rateprof)$ in the limit as $\varepsilon\downarrow 0$, providing a \ac{MES} as claimed.
\end{proof}

\begin{remark}
The quadratic regularizer $\varepsilon\|\loadprof\|^2$ was introduced by Tikhonov in the study of ill-posed inverse problems \citep{Tih:CRASURSS1943,Tih:DANSSSR1963,TihArs:WS1977}.
It is also the basis of \emph{ridge regression} in statistics \citep{Hoe:CEP1959,Hoe:CEP1962,HoeKen:T1970}.
In our setting this is just one choice among others, and can be replaced by a separable regularizer $\varepsilon\sum_{i=1}^ng_i(\load_i)$ with $g_i'(\argdot)$ strictly increasing. 
Every such regularizer selects a specific optimal solution in the limit when $\varepsilon\downarrow 0$ (see \citet[theorem 2.1]{Att:SIAMJO1996}
and
\citet[proposition 2.5]{AusComHad:MOR1997}.
Moreover, one can verify that the previous proof is still valid and yields a  monotone selection of the set of Wardrop equilibria. 
In particular,  $\varepsilon\sum_{i=1}^n \load_i\log(\load_i)$ selects the Wardrop equilibrium of maximal entropy. 
A similar entropic regularization was used in \citet{RosMcNHen:JUPD1989} to select one among multiple flow decompositions of a Wardrop equilibrium \citep[see][for a survey of related work]{BorBreKerSto:TRB2015}.
In our case we  deal with multiple equilibria and the regularization is used to obtain a selection with monotonicity properties.
As alternatives one may consider general penalty schemes of the form $\varepsilon\sum_{i=1}^n\theta(\load_i/\varepsilon)$, including the 
classical log-barrier $\theta(\load)=-\log(\load)$, the inverse-barrier $\theta(\load)=1/\load$, the exponential penalty $\theta(\load)=\exp(-\load)$, and more  \citep[see][]{Com:Springer1999}.
Let us also mention the multi-scale regularizer  $\sum_{i=1}^n \varepsilon^i \load_i^2$, which yields  a hierarchical selection principle: select the Wardrop equilibria that have the smallest first coordinate $\load_{1}^2$, among  these the ones with smallest $\load_{2}^2$, and inductively with $\load_3^2,\ldots,\load_n^2$.
\end{remark}

\begin{remark}
\cref{thm:singleton-congestion-games} is related to  results in  \citet{FujGoeHarPeiZen:MOR2017}, which investigates Braess's paradox in the context of \emph{nonatomic matroid congestion games}, where the strategy set for each commodity $\commodity$ is the set $\mathcal{B}^{\commodity}$ of bases of some matroid $M^{\commodity}=(\resources,\mathcal{I}^{\commodity})$, defined over a common ground set $\resources$ of resources. 
Among other results, lemma~3.2  in that paper establishes the monotonicity of the resource costs at equilibrium, from which one can readily deduce the monotonicity of the loads when the cost functions are strictly increasing. 
\end{remark}

%
%

\section{Comonotonicity and Active Regimes in Singleton Congestion Games}
\label{sc:parallel-regimes}

\cref{thm:singleton-congestion-games} shows that the equilibrium loads in  singleton congestion games respond monotonically when all the demands increase or stay the same. In mixed cases where some demands increase and others decrease, one can still identify groups of resources that behave comonotonically in specific regions of the space of demands.
A trivial example is when all  commodities
can use every resource $\resources^{\commodity}\equiv\resources$, so they can be treated as a single commodity and the equilibrium loads are just nondecreasing functions of the total demand $\rate_{\commodities}=\sum_{\commodity\in\commodities}\rate^{\commodity}$.
More generally, we will show that a subset $\costclass\subset\commodities$ of commodities that have the same equilibrium cost, behave as if they were a single-commodity on a smaller congestion game restricted to a  subset  $\resources_{\costclass}$ of  resources, and the equilibrium loads of these resources are nondecreasing functions of the aggregate demand $\rate_{\costclass}$ of the group, so that they are  comonotonic.

To state our result  precisely, given a singleton congestion game structure $\game=(\resources,\costprof,\strategiesprof)$, we partition the space of demands $\reals_+^{\commodities}$ into different regions $
\region^{\preceq}$ characterized by the order in which the commodities are ranked by equilibrium cost. In order to understand the geometry of  such regions, we further decompose them into sub-regions corresponding to different 
active regimes.

\begin{definition}
\label{def:subregions}
For any fixed weak order $\precsim$ on $\commodities$  we call  $\region^\precsim$ the set of demands
that rank the commodities exactly in this order, that is,  
\begin{equation}
\label{eq:Gamma-prec}
\region^\precsim=\braces*{\rateprof\in\Rpos^{\commodities} \colon \eqcost^{\commodity}(\rateprof)\leq \eqcost^{\commodityalt}(\rateprof) \iff \commodity\precsim\commodityalt \text{ for all } \commodity,\commodityalt\in\commodities},
\end{equation}
and we call $\region^{\precsim}_{\regimeprof}$ the \emph{sub-region with active regime}
$\regimeprof\coloneqq\parens*{\regime^{\commodity}}_{\commodity\in\commodities}$ with $\regime^{\commodity}\subset\resources^{\commodity}$, that is,
\begin{equation}
\label{eq:Gamma-rho-prec}    \region^{\precsim}_{\regimeprof}
\coloneqq\braces*{\rateprof\in \region^{\precsim} \colon
\activ{\resourcesprof}(\rateprof)=\regimeprof
}.
\end{equation}
\end{definition}

We recall that  the equivalence relation and strict order associated with $\precsim$ are  defined by
\begin{align*}
(\commodityalt \sim \commodity) & \text{ if and only if }
(\commodity \precsim \commodityalt)\text{ and } (\commodityalt \precsim \commodity),\\
(\commodityalt \succ \commodity) &\text{ if and only if }
(\commodity \precsim \commodityalt)\text{ and } \neg(\commodityalt \precsim \commodity).
\end{align*}
The relation $\sim$ partitions $\commodities$ into equivalence classes, called \emph{cost classes}: two commodities are in the same cost class if and only if $\commodity\sim\commodityalt$, that is to say, if and only if 
$\eqcost^{\commodity}(\rateprof)=\eqcost^{\commodityalt}(\rateprof)$ for all $\rateprof\in\region^{\precsim}$. 
To each cost class $\costclass$ we associate the subset $\resources_{\costclass}$ 
of all the resources $\resource\in\resources$ that are feasible for some commodity $\commodity\in\costclass$,  excluding
those which are also feasible for higher ranked commodities $\commodityalt\succ\commodity$, that is
\begin{equation}
\label{eq:R-C}
\resources_\costclass=\parens*{\cup_{\commodity\in\costclass}\resources^{\commodity}} \setminus \parens*{\cup_{\commodityalt\succ\costclass}\resources^{\commodityalt}}.
\end{equation}

\begin{definition}
\label{def:region-order}
 Let $\costclass$ be a cost class for a  weak order $\precsim$ on $\commodities$. 
We let
$\game_\costclass\coloneqq(\resources_\costclass,\costprof,\strategies_{\costclass} )$ denote the  singleton congestion game strucutre with a \emph{single commodity} whose strategy set $\strategies_{\costclass}$  comprises  all the singletons in $\resources_\costclass$.
\end{definition}

The regions $\region^\precsim$ can be empty for some orders $\precsim$ (\eg if $\commodity,\commodityalt\in\commodities$ are such that  $\resources^{\commodity}\subseteq\resources^{\commodityalt}$  we cannot have $\eqcost^{\commodity}(\rateprof)<\eqcost^{(\commodityalt)}(\rateprof)$).
We stress that each commodity $\commodity\in\commodities$ belongs to a unique cost class $\costclass$, whereas each resource $\resource$ belongs to the cost class of the highest ranked commodity among those for which $\resource$ is feasible.

\begin{example}
\label{ex:parallel-regions-and-sub-regions-0}

Consider a singleton congestion game structure with three resources $\resources=\{\resource_{1},\resource_{2},\resource_3\}$ 
with affine costs $\cost_{1}(x)=x+1$, $\cost_{2}(x)=x$, $\cost_3(x)=x+2$, and two commodities $\alpha$ and $\beta$ with $\resources^{\alpha}=\braces*{\resource_{1},\resource_{2}}$ and
$\resources^{\beta}=\braces*{\resource_{2},\resource_{3}}$.
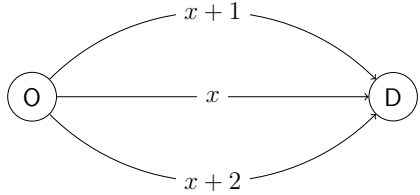
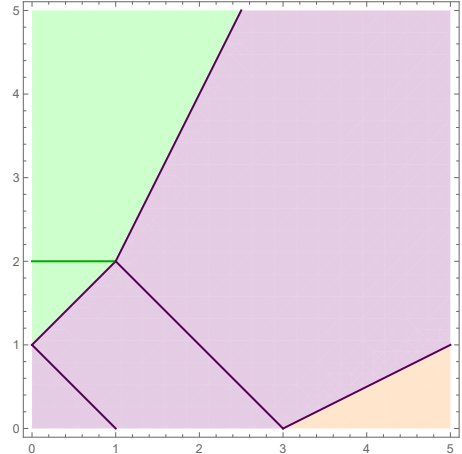
\begin{figure}[ht]
\hfil
\subfigure[A routing game with one commodity that uses the two top edges, and a second commodity that uses the bottom two.]{\raisebox{16mm}{
\scalebox{0.8}{\begin{tikzpicture}
   \node[shape=circle,draw] (v1) at (-3,0) {$\source$}; 
   \node[shape=circle,draw] (v6) at (3,0) {$\sink$}; 
   \draw[->] (v1) to [bend left=45] node[midway,fill=white] {$\load+1$} (v6);  
   \draw[->] (v1) to node[midway,fill=white] {$\load$} (v6);
   \draw[->] (v1) to [bend right=45] node[midway,fill=white] {$\load+2$} (v6);    
\end{tikzpicture}
}}
\label{fig:parallel-affine-graph}}
\hfil
\subfigure[The  colors  represent the regions $\region^\precsim$ for the three possible orders $\precsim$ of the equilibrium costs. 
These regions are further decomposed into polyhedral subregions $\region^{\precsim}_{\regime}$ that correspond to different active regimes.]
{\scalebox{0.9}{\includegraphics[width=0.4\textwidth]{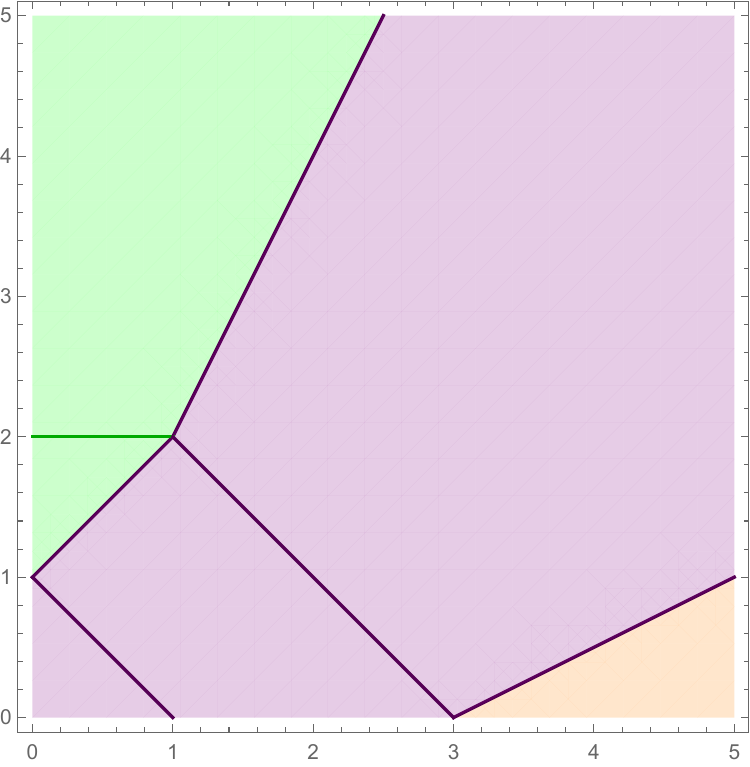}
\label{fig:parallel-affine-picture}}}
\caption{A singleton congestion game with affine costs.}
\label{fig:parallel-affine-0}
\end{figure}
For visualization, \cref{fig:parallel-affine-graph}  represents this  as a routing game on a parallel network, where both commodities have to move traffic between the two vertices, but each of them is allowed to use only certain edges.
In \cref{fig:parallel-affine-picture} the horizontal axis represents the demand of commodity $\alpha$ and the vertical axis the demand of $\beta$.
The three colors represent the regions $\region^\precsim$ corresponding to the possible orders of $\eqcost^{\alpha}$ and $\eqcost^{\beta}$.
In the top-left region in green $\eqcost^{\alpha}<\eqcost^{\beta}$, in the bottom-right region in orange $\eqcost^{\alpha}>\eqcost^{\beta}$, whereas in the middle region in purple $\eqcost^{\alpha}=\eqcost^{\beta}$. 
Hence, in the top-left and bottom-right regions we have two cost classes $\costclass_{1}=\{\alpha\}$ and $\costclass_{2}=\{\beta\}$, each containing one commodity. 
However, in the top-left region the corresponding resource sets
are $\resources_{\costclass_{1}}=\{\resource_{1}\}$ and $\resources_{\costclass_{2}}=\{\resource_{2},\resource_3\}$, whereas in the bottom-right region $\resources_{\costclass_{1}}=\{\resource_{1},\resource_{2}\}$ and $\resources_{\costclass_{2}}=\{\resource_3\}$. 
The region in purple has  a single cost class 
$\costclass=\{\alpha,\beta\}$ with $\resources_{\costclass}=\{\resource_{1},\resource_{2},\resource_3\}$.
The sub-regions delimited by horizontal and diagonal lines
within a colored region, correspond to different active regimes.
In the purple region characterized by   $\eqcost^{\alpha}=\eqcost^{\beta}$, with a single cost class $\costclass=\{\alpha,\beta\}$ and $\resources_{\costclass}=\{\resource_{1},\resource_{2},\resource_3\}$, there are three sub-regions depending on the value of the total demand $
\rate_{\costclass}=\rate^{\alpha}+\rate^{\beta}$. When
$\rate_{\costclass}\in(0,1)$ both $\alpha$ and $\beta$  use only the central edge with active regime $\regime^{\alpha}=\regime^{\beta}=\{\resource_{2}\}$ and equilibrium costs $\eqcost^{\alpha}=\eqcost^{\beta}=\rate_{\costclass}$. 
For $\rate_\costclass\in[1,3)$ we have $\eqcost^{\alpha}=\eqcost^{\beta}=(1+\rate_{\costclass})/2$,
with $\beta$ using the central edge $\regime^{\beta}=\{\resource_{2}\}$, whereas $\alpha$ splits the flow between the top and central edge with
$\regime^{\alpha}=\{\resource_{1},\resource_{2}\}$.
Finally for $\rate_\costclass\geq 3$ the active regime is 
$\regime^{\alpha}=\{\resource_{1},\resource_{2}\}$ and $\regime^{\beta}=\{\resource_{2},\resource_3\}$
with equilibrium costs $\eqcost^{\alpha}=\eqcost^{\beta}=1+\rate_{\costclass}/3$.
Similarly, the green region is characterized by $\eqcost^{\alpha}<\eqcost^{\beta}$ with cost classes $\costclass_{1}=\{\alpha\}$ and $\costclass_{2}=\{\beta\}$. Throughout this green region the active regime for $\alpha$ is constant $\regime^{\alpha}=\{\resource_{1}\}$, whereas 
$\regime^{\beta}=\{\resource_{2}\}$ if $\rate^\beta<2$ and $\regime^{\beta}=\{\resource_{2},\resource_3\}$ if $\rate^{\beta}\geq 2$.  
\end{example}

Our next result describes the equilibrium within a cost class $\costclass$: we show that the loads on the resources in $\resources_{\costclass}$ coincide with those of the single-commodity game $\game_{\costclass}$. In other words, in terms of equilibrium loads the commodities in $\costclass$ behave as if they were a single commodity.
This allows in turn to analyze the regions on which the equilibrium loads on  $\resources_\costclass$ are comonotone. 
Moreover, part \eqref{it:thm:parallel-regions-c} further analyzes the geometry of the regions of comonotonicity, as observed in \cref{ex:parallel-regions-and-sub-regions-0}.
The simple structure exhibited by the sub-regions in \cref{ex:parallel-regions-and-sub-regions-0} holds more generally: even if the cost functions are nonlinear, the sub-regions  are separated by hyperplanes defined by the aggregate demand of some cost class. 
We recall that a \emph{break point} in a single commodity game
is a demand $\bar\rate$ at which the set of active resources
changes, \ie this set is not constant
on any interval $(\bar\rate-\varepsilon,\bar\rate+\varepsilon)$ with $\varepsilon>0$ \citep[see][definition~3.4]{ComDosSca:MP2021}.

\begin{theorem}
\label{thm:parallel-regions}
Let $\game=(\resources,\costprof,\strategiesprof)$ be a singleton congestion game structure,
and $\region^{\precsim}$ the region  associated with a weak order $\precsim$ on $\commodities$. Then, for each  cost class $\costclass$ for $\precsim$ we have: 
\begin{enumerate}[\upshape(a)]
\item 
\label{it:thm:parallel-regions-a}
For all $\rateprof\in\region^\precsim$ and every equilibrium load $\loadprof$ of $(\game,\rateprof)$, the vector
$\bar\loadprof= (\load_\resource)_{\resource\in\resources_{\costclass}}$ is an equilibrium in the single-commodity game $(\game_\costclass,\rate_{\costclass})$ with aggregate demand 
$\rate_\costclass \coloneqq \sum_{\commodity\in{\costclass}} \rate^{\commodity}$.

\item 
\label{it:thm:parallel-regions-b}
If $\game_\costclass$ has a unique equilibrium for each demand in $\Rpos$,
then for $\rateprof\in\region^{\precsim}$ the equilibrium loads $\load_\resource(\rateprof)$ with $\resource\in\resources_\costclass$ can be expressed as nondecreasing functions of the aggregate demand $\rate_{\costclass}$,
which is equivalent to the fact that the equilibrium loads of the resources in $\resources_{\costclass}$ are comonotonic in the region $\region^\precsim$.

\item 
\label{it:thm:parallel-regions-c}
If the costs are strictly increasing, then 
the boundary between the sub-regions $\region^{\precsim}_{\regimeprof}$ coincides with the points $\rateprof\in\region^{\precsim}$ satisfying at least one of the linear equations
\[
\sum_{\commodity\in\costclass}\rate^{\commodity}=\bar\rate,
\]
where $\bar\rate$ is a break point in the single-commodity  game $\game_\costclass$.
\end{enumerate}
\end{theorem}

\begin{proof}
\eqref{it:thm:parallel-regions-a} 
Let $\loadprof$ be an equilibrium load vector of demand $\rateprof\in\region^\precsim$. We note that every commodity $\commodity\in\costclass$ allocates traffic only through resources in $\resources_\costclass$. 
Indeed, if a commodity $\commodity\in\costclass$ has a feasible resource also in $\resources_{\costclass'}$ with $\costclass'\ne\costclass$, then, because of \eqref{eq:R-C}, we have $\commodityalt\succ\commodity$ for every $\commodityalt\in\costclass'$, which is equivalent to $\eqcost^{\commodityalt}(\rateprof)>\eqcost^{\commodity}(\rateprof)$,  because $\rateprof\in\region^\precsim$. 
For this reason, all the commodities $\commodity\in\costclass$ have the same equilibrium cost $\eqcost^{\commodity}(\rateprof)=:\eqcost_{\costclass}(\rateprof)$, which implies that for every $\resource,\resourcealt\in\resources_{\costclass}$ we have
\begin{equation*}
\load_\resource>0\quad\implies
\quad\cost_\resource(\load_\resource)= \eqcost_{\costclass}(\rateprof)\le\cost_\resourcealt(\load_\resourcealt).
\end{equation*}
Since $\sum_{\resource\in\resources_{\costclass}}\load_\resource=\sum_{\commodity\in\costclass}\rate^{\commodity}=\rate_\costclass$, the vector $\bar\loadprof=(\load_\resource)_{\resource\in\resources_{\costclass}}$ is a single-commodity equilibrium for $\game_\costclass$ with demand $\rate_\costclass$. It follows that $\eqcost_{\costclass}(\rateprof)$ is in fact a function of the aggregate demand
$\rate_\costclass$ and so we can write it as
$\eqcost_{\costclass}(\rate_\costclass)$.
\medskip

\noindent
\eqref{it:thm:parallel-regions-b} 
By the result in \eqref{it:thm:parallel-regions-a}, for each $\resource\in\resources_\costclass$ and $\rateprof\in\region^\precsim$ the  equilibrium load $\load_\resource(\rateprof)$ coincides with the unique equilibrium in the single-commodity game $\game_{\costclass}$ with demand $\rate_{\costclass}$, and therefore it is a function of the aggregate demand $\rate_{\costclass}$. 
Now, according to \citep[proposition 3.12]{ComDosSca:MP2021}  every single-commodity game on a \ac{SP}  network has a nondecreasing selection of equilibria, so that $\load_\resource(\rateprof)$ is a nondecreasing function of $\rate_{\costclass}$. 
The equivalence with the comonotonicity of the maps $\rateprof\mapsto\load_\resource(\rateprof)$ for  $\resource\in\resources_{\costclass}$ throughout the region $\rateprof\in\region^\precsim$,
then follows from  a known result (see \eg \citet{Del:SPUS1970} and 
\citet{LanMei:AOR1994}).
Since we could not find a proof of this latter result in the literature, we include one  in \cref{lem:comonotonic} in \cref{sc:appendix:proofs}.
\medskip

\noindent
\eqref{it:thm:parallel-regions-c} Consider any demand $\rateprof\in\region^{\precsim}$.
By \eqref{it:thm:parallel-regions-a}, the equilibrium loads can be partitioned by cost classes $(\load_{\resource}(\rateprof))_{\resource\in\resources_{\costclass}}$, the latter being an equilibrium in the single-commodity game $\game_{\costclass}$. The equilibrium cost for $\game_{\costclass}$ is a strictly increasing function $\rate_{\costclass}\mapsto\eqcost_{\costclass}(\rate_{\costclass})$
of the aggregate demand 
$\rate_{\costclass}=\sum_{\commodity\in\costclass}\rate^{\commodity}$. It then follows that each load $\load_{\resource}(\rateprof)=\cost_{\resource}^{-1}(\eqcost_{\costclass}(\rate_{\costclass}))$ for $\resource\in\resources_{\costclass}$ is also a strictly increasing function of $\rate_{\costclass}$.

If $\rateprof\in\region^{\precsim}$ is on the boundary between two or more sub-regions $\region^{\precsim}_{\regimeprof}$,
the set of active resources changes locally at $\rateprof$ and then there must exist a cost class $\costclass$ whose set of active resources also changes locally at $\rate_{\costclass}$, which is therefore a break point in the single-commodity game $\game_{\costclass}$.
\end{proof}

\begin{example}
\label{ex:parallel-regions-and-sub-regions}

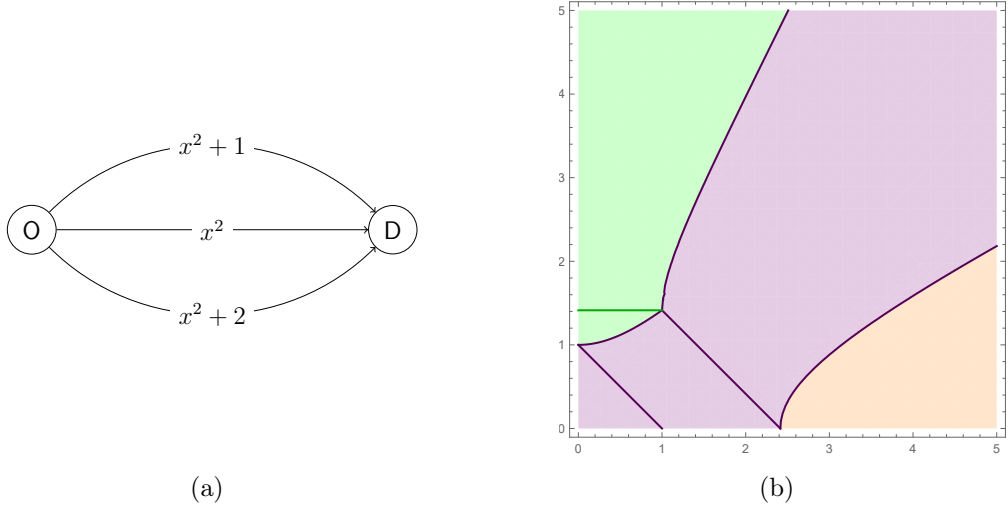
\begin{figure}[ht]
\hfil
\subfigure[]{\raisebox{16mm}{\scalebox{0.8}{
\begin{tikzpicture}
   \node[shape=circle,draw] (v1) at (-3,0) {$\source$}; 
   \node[shape=circle,draw] (v6) at (3,0)  {$\sink$};
   \draw[->] (v1) to [bend left=45] node[midway,fill=white] {$\load^2+1$} (v6);  
   \draw[->] (v1) to node[midway,fill=white] {$\load^2$} (v6);
   \draw[->] (v1) to [bend right=45] node[midway,fill=white] {$\load^2+2$} (v6);    
\end{tikzpicture}
}}}
\hfil
\subfigure[]
{
\scalebox{0.9}{\includegraphics[width=0.4\textwidth]{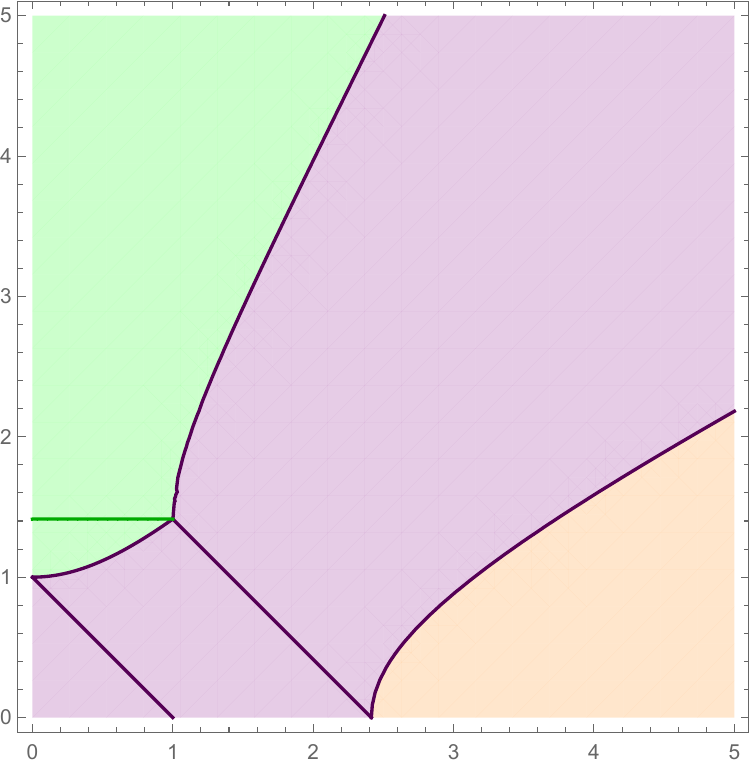}
}}
\caption{An example with quadratic costs. 
The first commodity uses the two top edges, and the second commodity uses the bottom two. 
The three colors represent the regions $\region^\precsim$ for the possible orders $\precsim$ of the equilibrium costs. The straight lines within each region separate  sub-regions corresponding to different active regimes. 
The regions $\region^{\precsim}$ are not convex, but the boundary between  sub-regions is still affine. 
}
\label{fig:parallel-BPR}
\end{figure}
Consider the variant of \cref{ex:parallel-regions-and-sub-regions-0} with quadratic costs as in \cref{fig:parallel-BPR}. 
The regions $\region^\precsim$ are no longer convex, but the sub-regions for  different active regimes are still delimited by the hyperplanes described in  \cref{thm:parallel-regions}\eqref{it:thm:parallel-regions-c}. 
In the purple region where $\eqcost^{\alpha}=\eqcost^{\beta}$ with a single cost class $\costclass=\{\alpha,\beta\}$ and $\resources_{\costclass}=\{\resource_{1},\resource_{2},\resource_3\}$,  the equilibrium loads of all the resources are strictly increasing functions of the total demand $\rate_{\costclass}=\rate^{\alpha}+\rate^{\beta}$, and the active regimes present
 break points at
$\rate_{\costclass}=1$ and $\rate_{\costclass}=1+\sqrt{2}$, that is, 
\begin{equation}
\label{eq:pho-alpha}
\begin{cases}
\regime^{\alpha}=\{\resource_{2}\}\text{ and }\regime^{\beta}=\{\resource_{2}\} & \text{if }~\rate_\costclass\in[0,1),\\
\regime^{\alpha}=\{\resource_{1},\resource_{2}\}\text{ and }\regime^{\beta}=\{\resource_{2}\} & \text{if }~ \rate_\costclass\in[1,1+\sqrt{2}),\\
\regime^{\alpha}=\{\resource_{1},\resource_{2}\} \text{ and } \regime^{\beta}=\{\resource_{2},\resource_3\} & \text{if }~
\rate_\costclass\in[ 1+\sqrt{2},\infty).    
\end{cases}    
\end{equation}

Similarly, in the green region where $\eqcost^{\alpha}<\eqcost^{\beta}$ with cost classes $\costclass_{1}=\{\alpha\}$ and $\costclass_{2}=\{\beta\}$, we have
$\regime^{\beta}=\{\resource_{2}\}$ for 
 $\rate^{\beta}<
\sqrt{2}$ and $\regime^{\beta}=\{\resource_{2},\resource_3\}$ for  $\rate^{\beta}\geq
\sqrt{2}$, whereas  $\regime^{\alpha}=\{\resource_{1}\}$ is constant.
\end{example}

\begin{remark}
\label{rem:comonotonicity_for_strictly_increasing_costs}
By \cref{rm:continuity-of-flows-when-costs-invertible}, having strictly increasing costs ensures the uniqueness of equilibria for $\game_\costclass$, as required in \cref{thm:parallel-regions}\eqref{it:thm:parallel-regions-b}. 
Actually, it suffices that no two resources in $\resources_\costclass$ have cost functions that are constant and equal on some (possibly different) non-degenerate intervals.
Moreover, for strictly increasing costs the equilibrium loads $\load_\resource(\rateprof)$ for $\resource\in\resources_\costclass$ and $\rateprof\in\region^{\precsim}$ are strictly increasing with $\rate_{\costclass}$. 
Indeed, since $\sum_{\resource\in\resources_{\costclass}}\load_\resource(\rateprof)=\rate_{\costclass}$, a strict increase of $\rate_{\costclass}$ implies that some load $\load_\resource(\rateprof)$ and its corresponding cost $\cost_\resource(\load_\resource(\rateprof))$ must strictly increase. 
However, across $\region^{\precsim}$ the equilibrium costs of all the resources $\resource\in\resources_{\costclass}$ remain equal, so that all their loads $\load_\resource(\rateprof)$ must   strictly increase simultaneously.
\end{remark}

\begin{remark}
\label{rem:comonotonicity_fails}
\cref{thm:parallel-regions}\eqref{it:thm:parallel-regions-b} implies that comonotonicity fails across different cost classes $\costclass\neq\costclass'$: if $\rate_{\costclass}$ increases and $\rate_{\costclass'}$ decreases, the  equilibrium loads of the resources $\resources_{\costclass}$ and $\resources_{\costclass'}$ will move in opposite directions. On the contrary, if both aggregate demands move in the same direction, the same holds for the corresponding equilibrium loads.
\end{remark}
\begin{remark}
The comonotonicity in \cref{thm:parallel-regions}\eqref{it:thm:parallel-regions-b} may  fail  when $\game_{\costclass}$ has multiple equilibria.
Consider for instance a variant of \cref{ex:parallel-regions-and-sub-regions-0} with costs $\cost_{1}(x)=\cost_3(x)=1$ and $\cost_{2}(x)=x$. 
When the demand is $\rateprof=(2,0)$ the  equilibrium sends 1 unit of flow through $\resource_{1}$ and $\resource_{2}$, and zero on $\resource_3$. 
Instead, at demand $\boldsymbol{\ratealt}=(0,2)$ nothing is sent through $\resource_{1}$, with 1 unit of traffic on both $\resource_{2}$ and $\resource_3$. 
Hence, despite the fact that at both $\rateprof$ and $\boldsymbol{\ratealt}$ all three resources have the same equilibrium cost equal to 1, the load on resource $\resource_{1}$ decreases when moving from  $\rateprof$ to $\boldsymbol{\ratealt}$, whereas the load on resource $\resource_3$ increases, so these loads are not comonotonic.
\cref{thm:parallel-regions}\eqref{it:thm:parallel-regions-b} does not apply here because the single-commodity game $\game_{\costclass}$ on the three resources and aggregate demand 2 has multiple equilibria.
\end{remark}

\begin{remark}
For single-commodity routing games on \ac{SP} networks the number of active regimes is at most the number of paths. 
This bound does not hold for multiple commodities and there can be as many 
as $\prod_{\commodity\in\commodities} \parens*{2^{\abs*{\resources^{\commodity}}}-1}$ potential combinations for $\activ{\resourcesprof}(\rateprof)$, (see \cref{App:number_of_active_regimes}). 
\end{remark}

%
%

\section{Beyond Singleton Congestion Games}
\label{sc:beyond-singleton}

In this section we provide some monotonicity results that go beyond the class of singleton congestion games studied in \cref{sc:singleton} and that also extend some known theorems for routing games with single \ac{OD} pair. 
We recall that in a standard routing the commodities coincide with \ac{OD} pairs and, 
moreover, the feasible strategies for each \ac{OD} pair are all the possible paths connecting the corresponding origin and destination.

\citet[proposition 3.12]{ComDosSca:MP2021} proved that in a single-\ac{OD} routing game over a \acl{SP} network the equilibrium load of each edge is nondecreasing in the traffic demand.
Every network that is not \acl{SP} contains a Wheatstone subnetwork \citep[see][]{Mil:GEB2006}; therefore, as shown in \cref{ex:Wheatstone}, there exist costs for which the equilibrium loads of some edges are decreasing in some demand interval.
This implies that the \acl{SP} nature  of the network is the best topological assumption that guarantees monotonicity of the equilibrium loads in a single-\ac{OD} setting. 

Unfortunately, for multi-\ac{OD} routing games the network topology alone does not  provide a criterion for the monotonicity of equilibrium loads. 
To obtain some useful results, we consider the following class of constrained routing games.

\begin{definition}
\label{def:constrained-routing-games}   
A \acfi{CRG}\acused{CRG}  is a tuple $(\graph,\commodities,\costprof,\routesprof,\rateprof)$ where 
\begin{itemize}
\item  
$\graph=(\vertices,\edges)$ is a directed multigraph with vertex set $\vertices$ and edge set $\edges$,

\item
$\commodities$
is a finite family of commodities,

\item 
$\costprof=(\cost_\edge)_{\edge\in\edges}$ is a vector of edge cost functions,

\item 
$\routesprof=(\routes^{\commodity})_{\commodity\in\commodities}$, 
with $\routes^{\commodity}$ a nonempty set of paths between an origin 
$\source^{\commodity}\in\vertices$ and a destination $\sink^{\commodity}\in\vertices$,

\item 
$\rateprof=\parens*{\rate^{\commodity}}_{\commodity\in\commodities}$ is a demand vector.
\end{itemize}

\end{definition}

This defines a congestion game structure 
with resource set $\resources=\edges$,
commodity set $\commodities$, costs 
$\costprof=(\cost_\edge)_{\edge\in\edges}$, and strategy sets $\strategies^{\commodity}=\routes^{\commodity}$.
Notice that in a constrained routing game the  commodities are distinguished by their different strategy sets $\routes^{\commodity}$, although they might share the same \ac{OD} pair and may also have some  paths in common. 
This is in contrast with standard routing games where each \ac{OD} pair is identified as a single commodity and $\routes^{\commodity}$ includes all the paths from $\source^{\commodity}$ to $\sink^{\commodity}$.
All the examples in \cref{sc:parallel-regimes} are in fact constrained routing games. 

Although restricting the  paths to a subset might seem a minor detail, it is in fact a flexible feature that allows us to represent any congestion game as a constrained routing game. 
Furthermore, we can also turn this routing game into a \emph{common-\ac{OD}} where all commodities  have the same origin and destination, by
\begin{itemize}
\item  
adding a super-source $\source$ connected to each  $\source^{\commodity}$ by a zero-cost edge $(\source,\source^{\commodity})$,
\item 
adding a super-sink 
$\sink$ connected by zero-cost edges $(\sink^{\commodity},\sink)$, and
\item 
appending the edges $(\source,\source^{\commodity})$ and
 $(\sink^{\commodity},\sink)$ to each path of commodity $\commodity$.
\end{itemize}

\noindent 
The next proposition shows  that every congestion game is equivalent 
to a common-\ac{OD} routing game over an extremely simple network, and all the complexity of the game is in fact encoded into the feasible sets of  paths.

Formally, two congestion game structures $\game$ and $\eqivgame\game$ are said to be \emph{equivalent} if there exist one-to-one correspondences $\commodity\leftrightarrow\eqivgame\commodity$ between their commodities and $\strategy\leftrightarrow\eqivgame \strategy$ between  strategies, such that for each  demand $\rateprof$ and each feasible flow $\flowprof$ of the first game, the flow  $\eqivgame\flowprof$ defined as $\eqivgame\flow_{\eqivgame\strategy}=\flow_{\strategy}$ is feasible in the second game and the strategy costs coincide $\eqivgame\cost_{\eqivgame\strategy}(\eqivgame\flowprof)=\cost_{\strategy}(\flowprof)$. 
In this case the equilibria of both games are also in one-to-one correspondence.

\begin{proposition}
\label{prop:congestion-to-routing}
Every congestion game is equivalent to a common-\ac{OD} constrained routing game over a \ac{SP} network.
\end{proposition}

\begin{proof}
Consider a congestion game structure with resources ${\resources}=\{\resource_{1},\dots,\resource_m\}$.
Consider the \ac{SP} network in the figure below,
\begin{figure}[ht]
\vspace{2ex}
\centering
\begin{tikzpicture}
\scalebox{1.1}{
    \node[shape=circle,draw=black,line width=.7pt] (v1) at (0,0)  {$\source$};
    \node[shape=circle,draw=black,line width=.7pt] (v2) at (2*1.3,0)  { }; 
    \node[shape=circle,draw=black,line width=.7pt] (v3) at (4*1.3,0)  { }; 
    \node[shape=circle,draw=black,line width=.7pt] (v4) at (7.0*1.3,0)  { }; 
    \node[shape=circle,draw=black,line width=.7pt] (v5) at (9*1.3,0)  {$\sink$ }; 

    \draw[line width=.7pt,->] (v1) to   node[midway,fill=white] {$\cost_{\resource_{1}}(\load)$} (v2);
    \draw[line width=.7pt,->] (v2) to   node[midway,fill=white] {$\cost_{\resource_{2}}(\load)$} (v3);
    \draw[line width=.7pt,->] (v1) to [bend right=60]  node[midway,fill=white] {$0$} (v2);
    \draw[line width=.7pt,->] (v2) to [bend right=60]  node[midway,fill=white] {$0$} (v3);
    \draw[line width=.7pt,->] (v3) to node[midway,fill=white] {$\dots$} (v4);
    \draw[line width=.7pt,->] (v3) to [bend right=45]  node[midway,fill=white] {$\dots$} (v4);
    \draw[line width=.7pt,->] (v4) to   node[midway,fill=white] {$\cost_{\resource_m}(\load)$} (v5);
    \draw[line width=.7pt,->] (v4) to [bend right=60]  node[midway,fill=white] {$0$} (v5);
   }
\end{tikzpicture}
\label{fig:congestion-to-network}
\end{figure}
where each resource is represented by two parallel edges: one of them has the original resource cost $\cost_\resource(\argdot)$, and the other edge provides a bypass with zero cost.
Any strategy $\strategy\subseteq {\resources}$ can be represented as a path joining $\source$ to $\sink$ that takes the top edge for each  resource in $\strategy$, and  
the bypass otherwise. 
We can then represent the commodities of the congestion game in the routing game by prescribing that they all have the same 
origin $\source$ and same destination $\sink$, whereas the feasible paths correspond to their feasible strategies in the original congestion game.
\end{proof}

Regarding the previous result, one may naturally ask whether a given nonatomic congestion game is equivalent to an \emph{unconstrained} nonatomic routing game. 
We are not aware of any result on this question, apart from the somewhat related result by \citet{Mil:IJGT2013}, who showed that every finite game can be represented as a \emph{weighted} atomic routing game.

As mentioned above, for standard multi-commodity routing games a \ac{SP} network topology does not suffice to guarantee the monotonicity of the equilibrium loads. 
Indeed, \cref{ex:Fisk-embedding,ex:braess-in-constrained-SP-form} below show that there exist common-\ac{OD} constrained routing games such that: 
\begin{itemize}
\item $\graph$ is \ac{SP};

\item every commodity uses paths $\routes^\commodity$ that form a \ac{SP} subnetwork;

\item the equilibrium loads $\loadprof(\rateprof)$ are unique; but

\item the map $\rateprof\mapsto\loadprof(\rateprof)$ is not a \ac{MES}.
\end{itemize}

\begin{example}\label{ex:Fisk-embedding}
Consider Fisk's network in \cref{subfig:Fisk} with $\cost_{\edge_{1}}(\load)=\cost_{\edge_{2}}(\load)=\load$, $\cost_{\edge{3}}(\load)=\load+90$, as in \cref{ex:Fisk}, and add bypass edges $\edge_{4},\edge_{5}$, as in \cref{subfig:Fisk-embedded}, with $\cost_{\edge_{4}}(\load)=\cost_{\edge_{5}}(\load)=0$, producing commodities $\commodity_{1}$, $\commodity_{2}$, $\commodity_3$ where $\source^\commodity=a$, $\sink^\commodity=c$ for every commodity $\commodity$, and  $\routes^{\commodity_{1}}=\{(\edge_{1},\edge_5)\}$, 
$\routes^{\commodity_{2}}=\{(\edge_4,\edge_{2})\}$, and $\routes^{\commodity_3}=\{(\edge_{1},\edge_{2}),\edge_3\}$. 
This defines an equivalent common-OD constrained routing game. 
As noted in \cref{ex:Fisk}, an increment in the demand of $\commodity_{1}$ 
pushes commodity $\commodity_3$ to divert more flow towards the direct path $\edge_3$, thus reducing the load 
on $\edge_{2}$  \citep[see][]{Fis:TRB1979}. 
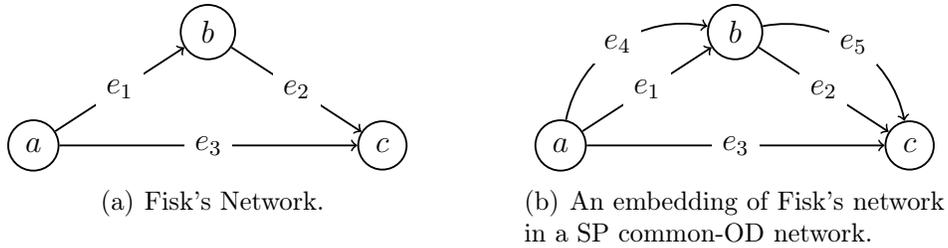
\begin{figure}[ht]
\centering
\setcounter{subfigure}{0}
\subfigure[Fisk's Network.]
{
\begin{tikzpicture}[scale=0.58]
    \node (e1) at (-1.25,2.2) {} ;
    \node (e2) at (1.25,2.2) {} ;
    \node (e3) at (2.2,0.3) {} ;
    \node[shape=circle,draw=black,line width=.7pt] (v1) at (-4,0)  {$a$}; 
   \node[shape=circle,draw=black,line width=.7pt] (v2) at (0,2.6)  { $b$}; 
   \node[shape=circle,draw=black,line width=.7pt] (v6) at (4,0)  { $c$}; 
   \draw[line width=.7pt,->] (v1) to   node[midway,fill=white] {$\edge_{1}$} (v2);
   \draw[line width=.7pt,->] (v1) to   node[midway,fill=white] {$\edge_{3}$} (v6);
   \draw[line width=.7pt,->] (v2) to   node[midway,fill=white] {$\edge_{2}$} (v6);
\end{tikzpicture}
\label{subfig:Fisk}
}
\hspace{1cm}
\subfigure[An embedding of Fisk's network in a \ac{SP} common-\ac{OD} network.]
{
\begin{tikzpicture}[scale=0.58]
    \node[shape=circle,draw=black,line width=.7pt] (v1) at (-4,0)  {$a$};
    \node (e1) at (-1.25,2.2)  {};
    \node (e2) at (1.25,2.2)  {};
    \node (e3) at (2.2,0.3)  {};
    \node (e4) at (-1.5,3.1)  {}; 
    \node (e5) at (1.5,3.1)  {};
   \node[shape=circle,draw=black,line width=.7pt] (v2) at (0,2.6)  { $b$}; 
   \node[shape=circle,draw=black,line width=.7pt] (v6) at (4,0)  { $c$}; 
   \draw[line width=.7pt,->] (v1) to   node[midway,fill=white] {$\edge_{1}$} (v2);
   \draw[line width=.7pt,->] (v1) to   node[midway,fill=white] {$\edge_{3}$} (v6);
   \draw[line width=.7pt,->] (v2) to   node[midway,fill=white] {$\edge_{2}$} (v6);
   \draw[line width=.7pt,->] (v1) to [bend left=45]  node[midway,fill=white] {$\edge_{4}$} (v2);
   \draw[line width=.7pt,->] (v2) to [bend left=45]  node[midway,fill=white] {$\edge_{5}$} (v6);
\end{tikzpicture}
\label{subfig:Fisk-embedded}}
\caption{Fisk's multi-commodity network can be embedded in a \ac{SP} network with a common-\ac{OD}, by adding two edges with zero cost.
}
\label{fig:Fisk-SP-embedding}
\end{figure}
\end{example}

\begin{example}\label{ex:braess-in-constrained-SP-form}
Monotonicity can also fail in a common-OD constrained routing game, even on a \ac{SP} graph. 
Indeed, the standard Braess's routing game in  \cref{fig:classic_braess} corresponds to the single-commodity congestion game structure 
$(\edges,\costprof,\{\source\,v _{1}\,v _{2}\sink, \source\,v _{1}\sink,\source\,v _{2}\sink\})$.
Using \cref{prop:congestion-to-routing} this is equivalent  to a  common-OD  constrained routing game on a \ac{SP}  network, for which the \ac{MES} property fails.
\end{example}

These examples show that, in addition to a \ac{SP} topology, we need to impose further conditions on how the commodities overlap. 
To this end we introduce the following operations of series and parallel connection of congestion game structures.

\begin{definition}
\label{def:CSP}
Let $\game_{1}=(\resources_{1},\costprof_{1},\strategiesprof_{1})$ and $\game_{2}=(\resources_{2},\costprof_{2},\strategiesprof_{2})$ be two congestion game structures with disjoint resource sets $\resources_{1}\cap\resources_{2}=\varnothing$. 
The series and parallel game structures are both defined on the resource set $\resources=\resources_{1}\cup\resources_{2}$
with their original cost functions. Specifically:
\begin{itemize}
\item 
The \emph{series} game structure $\game_{1} \times \game_{2}$ has commodities
 $(\commodity_1,\commodity_2)\in\commodities_{1}\times\commodities_{2}$, with corresponding 
 strategy set $\strategies^{(\commodity_1,\commodity_2)}=\{\strategy_{1}\cup \strategy_{2}: (\strategy_{1},\strategy_{2})\in\strategies_{1}^{\commodity_1}\times\strategies_{2}^{\commodity_2}\}$.
 
\item 
The \emph{parallel} game structure  $\game_{1}\cup \game_{2}$ has commodity set $\commodities=\commodities_{1}\cup\commodities_{2}$ and the original strategy sets $\strategies_{1}^{\commodity_1}$ for  $\commodity_1\in\commodities_{1}$ and $\strategies_{2}^{\commodity_2}$ for $\commodity_2\in\commodities_{2}$.

\item 
A \acfi{CSP}\acused{CSP} congestion game structure is constructed starting from singleton congestion game structures and applying a finite number of series or parallel connections to game structures already constructed.
\end{itemize}
\end{definition}

Whereas the parallel connection is a simple superposition of disjoint commodities, its combination with the series connection and the possibility of imposing constraints in the set of resources, provides a flexible tool to distinguish different types of commodities and to represent complex strategy sets that result from sequential processes.
Consider for instance a family of different job classes, each one representing a commodity $\commodity\in\commodities$, which must  be processed in a series of stages $k\in K$. 
At every stage there is a set of  machines $M_k$ that work in parallel to perform the given task, while the jobs of type $\commodity$ can only be processed in a subset $M_{k}^{\commodity}\subset M_k$.
A commodity $\commodity$ can then be identified with a particular sequence of feasible machines $(M^{\commodity}_k)_{k\in K}$, and its strategy set  corresponds to the strategy set for a connection in series  of singleton congestion game structures, one for each stage.
On the other hand, some job classes might not require some processing stages, which can be modeled as a bypass strategy  
using the parallel operation. 
As an illustration, the graph in \cref{fig:series-parallel-graph} can represent simultaneously commodities that must go 
sequentially over all four processing stages, possibly with restrictions on the machines that are allowed for each of them, as well as commodities that
only perform the first and fourth stages, and skip the two intermediate stages. 
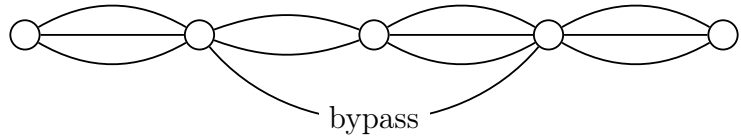
\begin{figure}[ht]
\centering
\begin{tikzpicture}[scale=0.58]
\node[shape=circle,draw=black,line width=.7pt] (a) at (-4,0)  {};
\node[shape=circle,draw=black,line width=.7pt] (b) at (0,0)  {}; 
\node[shape=circle,draw=black,line width=.7pt] (c) at (4,0)  {}; 
\node[shape=circle,draw=black,line width=.7pt] (d) at (8,0)  {}; 
\node[shape=circle,draw=black,line width=.7pt] (e) at (12,0)  {}; 
\draw[line width=.7pt] (a) to (b);
\draw[line width=.7pt] (a) to [bend left=30] (b);
\draw[line width=.7pt] (a) to [bend right=30]  (b);
\draw[line width=.7pt] (b) to [bend left=20] (c);
\draw[line width=.7pt] (b) to [bend right=20]  (c);
\draw[line width=.7pt] (c) to (d);
\draw[line width=.7pt] (c) to [bend left=30] (d);
\draw[line width=.7pt] (c) to [bend right=30]  (d);
\draw[line width=.7pt] (d) to (e);
\draw[line width=.7pt] (d) to [bend left=30] (e);
\draw[line width=.7pt] (d) to [bend right=30]  (e);  
\draw[line width=.7pt] (b) to [bend right=50] node[midway,fill=white] {bypass} (d);
\end{tikzpicture}
\caption{The \ac{SP} graph for a job-processing game.}
\label{fig:series-parallel-graph}
\end{figure}
The bypass strategy can also be
replaced by a series of alternative processing stages. 
Using such series and parallel operations one can model complex processing paths for different job classes. This construction gives rise to an \ac{SP} graph, however the main additional ingredient is which combinations of commodities are allowed along the construction.

\begin{theorem}
\label{pr:constrained-series-parallel}
Every \ac{CSP} congestion game structure has a \ac{MES}.
\end{theorem}

\begin{proof}
By induction and \cref{thm:singleton-congestion-games}, it suffices to show that the \ac{MES} property is preserved under series and parallel operations on game structures.
To this end, let $\game_{1}$ and $\game_{2}$ be two congestion game structures with \ac{MES}'s
$\rateprof_{1}\mapsto\loadprof_{1}(\rateprof_{1})$
and $\rateprof_{2}\mapsto\loadprof_{2}(\rateprof_{2})$ respectively. Then we prove the two parts:

\vspace{1ex}
\begin{enumerate}[(a)]
\item
\emph{The series game structure $\game_{1}\times \game_{2}$ has a \ac{MES}}. 
Let $\rateprof=(\rate^{(\commodity_1,\commodity_2)})_{\commodity_1\in\commodities_{1},\commodity_2\in\commodities_{2}}$ with $\rate^{(\commodity_1,\commodity_2)}$ the demand for the commodity  $(\commodity_1,\commodity_2)$ in a  game whose structure is $\game_{1}\times \game_{2}$. 

Define
\begin{equation}
\label{eq:demand-series}
\forall \commodity_1\in\commodities_{1}, \quad \rate_{1}^{\commodity_1}=\sum_{\commodity_2\in\commodities_{2}} \rate^{(\commodity_1,\commodity_2)};\quad \forall \commodity_2\in\commodities_{2}, \quad
\rate_{2}^{\commodity_2}=\sum_{\commodity_1 \in\commodities_{1}}\rate^{(\commodity_1,\commodity_2)},
\end{equation}
$\rateprof_{1} = \parens*{\rate_{1}^{\commodity_{1}}}_{\commodity_{1}\in\commodities_{1}}$, and $\rateprof_{2} = \parens*{\rate_{2}^{\commodity_{2}}}_{\commodity_{2}\in\commodities_{2}}$.
An equilibrium for $\rateprof$ can be obtained by superposing $\loadprof_{1}(\rateprof_{1})$ on the resources $\resources_{1}$ and $\loadprof_{2}(\rateprof_{2})$ on the resources $\resources_{2}$.
Since an increase of any demand $\rate^{(\commodity_1,\commodity_2)}$ induces an increase in the demands $\rate_{1}^{\commodity_1}$ and $\rate_{2}^{\commodity_2}$, the loads in $\loadprof_{1}(\rateprof_{1})$  and $\loadprof_{2}(\rateprof_{2})$ increase, so that this superposed equilibrium provides a \ac{MES} for the series game.

\item \emph{The union game structure $\game_{1}\cup \game_{2}$ has a \ac{MES}}.
For each demand $\rateprof=(\rateprof_{1},\rateprof_{2})$ in a game with structure $\game_{1}\cup \game_{2}$ we can directly find an equilibrium by superposing $\loadprof_{1}(\rateprof_{1})$ on the resources $\resources_{1}$ and $\loadprof_{2}(\rateprof_{2})$ on the resources $\resources_{2}$.
Since the loads in these two equilibria are monotone with respect to each individual demand, the same holds for their superposition which provides a  \ac{MES} for the parallel game structure. \qedhere
\end{enumerate}
\end{proof}

\begin{remark}
The common-OD constrained routing game of \cref{ex:Fisk-embedding}, in which Fisk's network is embedded, does not have a \ac{CSP} structure:
the strategy set of $\commodity_3$ cannot be obtained  as a strategy set of a previously present commodity when constructing a parallel game structure. 
For a different reason, the common-\ac{OD} \ac{SP} routing game in \cref{ex:braess-in-constrained-SP-form} does not have a \ac{CSP} game structure either. 
Indeed, its network would be made of 5 two-edge parallel network in series, each one associated to a resource, \ie an edge of the Wheatstone network. 
The classical Braess's routing game has a single commodity with three strategies, and a strategy set with cardinality 3 cannot be obtained as the Cartesian product of the strategy sets of the 5 two-edge parallel networks connected in series. 
A series of Pigou games has strong limitations, for example, a commodity with a number of paths divisible by a prime larger than 2 is not constructible as in \cref{def:CSP}. 

\end{remark}

\begin{remark}
It is somewhat odd that the class of \ac{CSP}'s does not include single-OD routing games over an \ac{SP} graph, in which there is a single commodity and every path is allowed. This happens because the parallel connection does not allow to merge
commodities and they are kept separated. 
However, if in the construction of \ac{CSP}'s instead of taking singleton congestion games as the initial atoms, one replaces each singleton strategy with a series-parallel routing game structure with a single commodity, then  \cref{pr:constrained-series-parallel} remains true for this larger class which trivially includes series-parallel routing games. 
\end{remark}

\begin{remark} Every \ac{CSP} game structure as defined above is a matroid game. 
In fact, if $\game_{1}$ and $ \game_{2}$ are two matroid game structures, then the parallel game structure $\game_{1}\cup \game_{2}$ is trivially a matroid game structure, whereas the strategy sets $\strategies^{(\commodity_1,\commodity_2)}$ in the series game structure $\game_{1}\times \game_{2}$ correspond to the direct sum operation on the original matroid bases $\strategies_{1}^{\commodity_1}$ and $\strategies_{2}^{\commodity_2}$.
As a consequence, for strictly increasing costs the previous result can also be derived from \citet[lemma~3.2]{FujGoeHarPeiZen:MOR2017}.    
\end{remark}

Whereas \ac{CSP} congestion game structures  were defined for general congestion games, they can also be described as constrained routing games with a specific structure. 
The following representation is also more natural compared to the one in \cref{prop:congestion-to-routing}.

\begin{theorem}\label{pr:constrained-series-parallel-routing-description}
Every \ac{CSP} congestion game structure $\game=(\resources,\costprof,\strategiesprof)$ is equivalent to a common-OD constrained routing game structure $(\graph,\costprof,\routesprof)$ such that
\begin{enumerate}[\upshape(i)]
\item 
\label{it:graph-series-parallel} 
the graph $\graph$ is \ac{SP},

\item \label{it:sequence-vertices} for each commodity $\commodity$ all the paths in $\routes^\commodity$ visit the same vertices in the same order,

\item 
\label{it:unsplittable-commodity}
for every two paths $\route_{1},\route_{2}\in\routes^\commodity$
and edges $\edge_{1}\in\route_{1}$ and $\edge_{2}\in\route_{2}$  connecting two subsequent vertices, the paths obtained from $\route_{1}$ and $\route_{2}$ by exchanging $\edge_{1}$ with $\edge_{2}$ also belong to $\routes^\commodity$.
\end{enumerate}
Furthermore, every common-OD constrained routing game satisfying \eqref{it:graph-series-parallel}, \eqref{it:sequence-vertices}, \eqref{it:unsplittable-commodity}, has a \ac{CSP} congestion game
structure.
\end{theorem}

\begin{proof}
Every \ac{CSP} game structure is built starting with singleton congestion games and applying a finite number of series or parallel operations. 
We start by noting that every singleton congestion game is equivalent to a constrained routing game on a parallel network with two vertices connected by edges corresponding to the resources of the singleton congestion game, and any such game satisfies \eqref{it:graph-series-parallel}, \eqref{it:sequence-vertices}, \eqref{it:unsplittable-commodity}. 
Hence, it suffices to show that these properties are preserved under series and parallel operations.

Consider two congestion game structures $\game_{1}=(\resources_{1},\costprof_{1},\strategiesprof_{1})$ and $\game_{2}=(\resources_{2},\costprof_{2},\strategiesprof_{2})$ which are  respectively equivalent to some common-OD constrained routing games $(\graph_{1},\costprof_{1},\routesprof_{1})$ and $(\graph_{2},\costprof_{2},\routesprof_{2})$ satisfying 
\eqref{it:graph-series-parallel}, \eqref{it:sequence-vertices}, \eqref{it:unsplittable-commodity}. 

The series game structure $\game_{1}\times \game_{2}$ is then equivalent to the constrained routing game structure $(\widetilde \graph,\widetilde \costprof, \widetilde \routesprof)$ where $\widetilde\graph$ is obtained by joining in series the graphs $\graph_{1}$ and $\graph_{2}$, the costs $\widetilde \costprof$ are the cost functions given by $\costprof_{1}$ and $\costprof_{2}$ on the corresponding edges, and the commodities are given by the sets of paths obtained choosing commodities $\commodity_{1}$ for $(\graph_{1},\costprof_{1},\routesprof_{1})$ and $\commodity_{2}$ for $(\graph_{2},\costprof_{2},\routesprof_{2})$, and joining every path in $\routes^{\commodity_{1}}$ with every path in $\routes^{\commodity_{2}}$ to construct paths in $\widetilde \graph$. 
Moreover, $(\widetilde \graph,\widetilde \costprof, \widetilde \routesprof)$ satisfies \eqref{it:graph-series-parallel}, \eqref{it:sequence-vertices}, \eqref{it:unsplittable-commodity} because  $(\graph_{1},\costprof_{1},\routesprof_{1})$ and $(\graph_{2},\costprof_{2},\routesprof_{2})$ do.

Similarly, the parallel game structure $\game_{1}\cup\game_{2}$ is equivalent to the constrained routing game structure $(\bar \graph,\bar \costprof, \bar \routesprof)$ where $\bar\graph$ is obtained joining in parallel $\graph_{1}$ and $\graph_{2}$, the costs $\bar \costprof$ are the cost functions given by $\costprof_{1}$ and $\costprof_{2}$ on the corresponding edges, and the commodities are given by the commodities of $(\graph_{1},\costprof_{1},\routesprof_{1})$ and $(\graph_{2},\costprof_{2},\routesprof_{2})$. 
Also in this case, the routing game $(\bar \graph,\bar \costprof, \bar \routesprof)$ satisfies \eqref{it:graph-series-parallel}, \eqref{it:sequence-vertices}, \eqref{it:unsplittable-commodity} because  $(\graph_{1},\costprof_{1},\routesprof_{1})$ and $(\graph_{2},\costprof_{2},\routesprof_{2})$ do.

This completes the proof of the first claim of the theorem. 

Conversely, notice that every \ac{SP} graph $\graph$ is constructed starting with parallel networks and joining them in series or in parallel for a finite number of times. 
Suppose that a common-\ac{OD} constrained routing game $(\graph,\costprof,\routesprof)$ structure satisfies \eqref{it:graph-series-parallel}, \eqref{it:sequence-vertices}, \eqref{it:unsplittable-commodity}.

If the graph $\graph$ is obtained by joining in series two graphs $\graph_{1}$ and $\graph_{2}$, we can endow them with cost functions which associate costs to edges as in $\costprof$. 
Furthermore, given a commodity $\commodity$ for $(\graph,\costprof,\routesprof)$ we can define commodities $\commodity_{1}$ on $\graph_{1}$ and $\commodity_{2}$ on $\graph_{2}$ by determining for $i=1,2$ the set of paths
\begin{equation*}
\routes^{\commodity_i}=\braces*{\route\text{ path in }\graph_i\text{ s.t. } \route\text{ is part of a path in }\routes^{\commodity}}.    
\end{equation*}
Since $(\graph,\costprof,\routesprof)$ satisfies \eqref{it:unsplittable-commodity}, we have $\routes^{\commodity}=\routes^{\commodity_{1}}\times\routes^{\commodity_{2}}$, so that $(\graph,\costprof,\routesprof)$ is the series game structure of the two constrained routing games just defined on $\graph_{1}$ and $\graph_{2}$.

If the graph $\graph$ is obtained by joining in parallel two graphs $\graph_{1}$ and $\graph_{2}$, then we can assume that the direct edges from the origin and the destination of $\graph$ are all contained in one of the two. 
We can again endow $\graph_{1}$ and $\graph_{2}$ with cost functions which associate costs to edges as in $\costprof$. 
Furthermore because of property \eqref{it:sequence-vertices}, for every commodity $\commodity$ of $(\graph,\costprof,\routesprof)$ the paths in $\routes^{\commodity}$ all belong to one between $\graph_{1}$ and $\graph_{2}$.
This allows us to define for each commodity of $(\graph,\costprof,\routesprof)$, a commodity either in $\graph_{1}$ or $\graph_{2}$, so that $(\graph,\costprof,\routesprof)$ is the parallel game structure of the two constrained routing games just defined on $\graph_{1}$ and $\graph_{2}$.
\end{proof}

\begin{remark}
Note that Braess's classical example in \cref{fig:classic_braess} satisfies \eqref{it:sequence-vertices} and  \eqref{it:unsplittable-commodity}, but does not satisfy \eqref{it:graph-series-parallel}. 
Fisk's network embedding of \cref{ex:Fisk-embedding} satisfies \eqref{it:graph-series-parallel} and \eqref{it:unsplittable-commodity} but does not satisfy  \eqref{it:sequence-vertices}. 
Finally, the constrained routing game of  \cref{ex:braess-in-constrained-SP-form}, obtained by embedding Braess's game in a \ac{SP} graph as in \cref{prop:congestion-to-routing}, satisfies \eqref{it:graph-series-parallel} and \eqref{it:sequence-vertices}, but not \eqref{it:unsplittable-commodity}.
 
\end{remark}

\begin{remark}
Conditions
\eqref{it:graph-series-parallel}, \eqref{it:sequence-vertices}, \eqref{it:unsplittable-commodity}
in \cref{pr:constrained-series-parallel-routing-description}
can be  equivalently stated by requiring that all feasible paths for a commodity $\commodity$ visit a specific ordered sequence of nodes; between successive nodes only a specific subset of parallel edges are allowed; and $\routes^\commodity$ includes all possible paths in this subnetwork. 
Still another equivalent description is to require that 
 for any two paths 
 $\route_{1},\route_{2}\in\routes^\commodity$ the mixed path where we follow $\route_{1}$ up to an intermediate node and then continue with $\route_{2}$ is also in $\routes^\commodity$. 
\end{remark}

%
%

\section{Summary and open problems}
\label{sc:summary}

This paper studied the monotonicity of  equilibrium travel times and equilibrium  loads in response to variations of the demands, identifying  conditions under which  the  paradoxical phenomena of non-monotonicity cannot happen.
We considered the general setting of congestion games,
with a special focus on singleton congestion games with multiple commodities for which we 
established in \cref{thm:singleton-congestion-games} the existence of a selection of the equilibrium loads which monotonically increase with respect to the demand of every commodity.

We next explored the notion of comonotonicity, which captures the idea that different resource loads jointly increase or decrease after variations of the demands.
\cref{thm:parallel-regions} described how comonotonicity is connected to the structure of equilibria in terms of how the commodities are ranked by cost and how the resources become active or inactive as the demands vary. We complemented this finding by a structural result on the regions of the demand space for which the same sets of resources are used at equilibrium. 

\cref{pr:constrained-series-parallel} extended the study of monotonicity from singleton congestion games to the larger class of congestion games having a \ac{CSP} structure, reminiscent of the concept of a \ac{SP} network.
We also derived an embedding that maps congestion games into constrained routing games (see \cref{prop:congestion-to-routing}) and characterized the classes of congestion games with good monotonicity properties by embedding them into routing games (see \cref{pr:constrained-series-parallel-routing-description}).
This last result sheds light on the features that produce the paradoxes and showcases the difference between  single and multiple \ac{OD} networks.  
When the network has a single \ac{OD} pair, its topology is the sole relevant factor to guarantee the monotonicity of equilibrium loads.
In the multiple \ac{OD} case the structure of the  paths that are in  each \ac{OD} pair also plays a crucial role.

A first open question not addressed in this paper, and which will be interesting to explore,
is how the structural results on the regions $\region^{\precsim}$ and sub-regions $\region^{\precsim}_{\regime}$ for the different active regimes might be exploited to devise an algorithm for building a curve of equilibria along a demand curve, analog to the path-following method for  piece-wise affine costs developed by \citet{KliWar:MOR2022}. A basic question here is to investigate the
geometry of the
regions $\region^{\precsim}$
for specific classes of cost functions.  For the special case  of \ac{BPR} costs, we conjecture that the boundaries between these regions are asymptotic to straight lines through the origin. 
This would imply that when the demands are scaled proportionally, the regimes will not repeat and  the curve will eventually enter into a particular asymptotic region
$\region^{\precsim}$ and remain there
forever. The latter could inspire a  path following algorithm to build a curve of equilibria.

A second open problem is to find an algorithm to recognize \ac{CSP} congestion game structures.
In this regard, one could be tempted to use the equivalent game in \cref{prop:congestion-to-routing} for 
which \eqref{it:graph-series-parallel} and \eqref{it:sequence-vertices} in \cref{pr:constrained-series-parallel-routing-description}
hold trivially, so that only \eqref{it:unsplittable-commodity} would need to be checked. 
Unfortunately, the \ac{CSP} property is not preserved under  equivalence: for instance, a singleton congestion game with only one commodity is  \ac{CSP}  by definition, but its equivalent representation in \cref{prop:congestion-to-routing} is not because property \eqref{it:unsplittable-commodity} fails.
This suggests that recognizing \ac{CSP} game structures is not straightforward. 
As a possible starting point
to address this question, one might try to adapt the existing algorithms for recognizing \ac{SP} networks \citep[see][]{ValTarLaw:SIAMJC1982,HeYes:IC1987,Epp:IC1992}.

\subsection*{Acknowledgments}
We thank the reviewers for their careful reading and  insightful comments. 
We also thank our colleague Tobias Harks for pointing out the connection of our results with the paper by \citet{FujGoeHarPeiZen:MOR2017}, and Tzachi Gilboa and Ludger R\"uschendorf for some historical insights about comonotonicity.
Valerio Dose and Marco Scarsini are members of GNAMPA-INdAM. 
Their work was partially supported by the GNAMPA project CUP\_E53C22001930001  ``Limiting behavior of stochastic dynamics in the Schelling segregation model'' and by the MIUR PRIN project 2022EKNE5K   ``Learning in Markets and Society.''
Roberto Cominetti's research was supported by Proyecto Anillo ANID/PIA/ACT192094.

\appendix

\gdef\thesection{\Alph{section}} 
\makeatletter
\renewcommand\@seccntformat[1]{\appendixname\ \csname the#1\endcsname.\hspace{0.5em}}
\makeatother

%
%

\section{Supplementary proofs}
\label{sc:appendix:proofs}

\subsection{Missing proof}
\begin{proof}[Proof of \cref{pr:continuity}]
Let $(\resources,\costprof,\strategies)$ be a nonatomic congestion game structure. 
For every demand $\rateprof\in\Rpos^{\commodities}$, let $\valueW(\rateprof)$ be the minimum value of the Beckmann potential as in \eqref{eq:Beckmann}, that is,
\begin{equation}
\label{eq:min-Beckmann}
\valueW(\rateprof)=\min_{\loadprof\in\loads_{\rateprof}}\sum_{\resource\in\resources}\Cost_\resource(\load_\resource).
\end{equation}
We obtain the result as a consequence of convex duality. Consider the function
$\perturb_{\rateprof}: \reals^{\strategiesprof}\times \reals^{\commodities}\to\reals\cup\{+\infty\}$ given by
\begin{equation}
\label{eq:perturb}
\perturb_{\rateprof}(\flowprof,\zvarprof)
=
\begin{cases}
\sum_{\resource\in\resources}\Cost_{\resource}\parens*{\sum_{\strategyalt\ni\resource}\flow_{\strategyalt}} & \text{if } \flowprof\geq \zerovec,\ \sum_{\strategy\in\strategies^{\commodity}}\flow_{\strategy}=\rate^\commodity+\zvar^{\commodity}\text{ for every }\commodity\in\commodities,\\
+\infty& \text{otherwise},
\end{cases}
\end{equation}
which is a proper closed convex function.
Letting $\vinf_{\rateprof}(\zvarprof)$ denote the optimal value function of the primal problem 
\begin{equation}
\label{eq:primal}
(\prim_{\rateprof})\quad\inf_{\flowprof}\perturb_{\rateprof}(\flowprof,\zvarprof),
\end{equation}
we have
$\valueW(\rateprof+\zvarprof)=\vinf_{\rateprof}(\zvarprof)$ and, in particular, $\valueW(\rateprof)=\vinf_{\rateprof}(\zerovec)$.

Since $\perturb_{\rateprof}$ is convex, we have that $\zvarprof\mapsto \vinf_{\rateprof}(\zvarprof)=\valueW(\rateprof+\zvarprof)$ is also convex,  from which we deduce that $\rateprof\rightarrow\valueW(\rateprof)$ is convex.
Moreover, the perturbed function $\perturb_{\rateprof}$ yields a corresponding dual 
\begin{equation}
\label{eq:dual}
(\dual_{\rateprof})\quad\min_{\eqcost\in\R^{\commodities}}\perturb_{\rateprof}^{*}(\zerovec,\eqcostprof),
\end{equation}
where $\perturb_{\rateprof}^{*}$ is the Fenchel conjugate function, that is,
\begin{equation}
\label{eq:phi-star}
\begin{split}
\perturb_{\rateprof}^{*}(\zerovec,\eqcostprof)
&=\sup_{\flowprof,\zvarprof}~\langle \zerovec,\flowprof\rangle+\langle\eqcostprof, \zvarprof\rangle-\perturb_{\rateprof}(\flowprof,\zvarprof)\\
&=\sup_{\flowprof\geq \zerovec} \sum_{\commodity\in\commodities}\parens*{\eqcost^{\commodity}\parens*{\sum_{\strategy\in\strategies^{\commodity}} \flow_{\strategy}-\rate^\commodity}}-\sum_{\resource\in\resources}\Cost_{\resource}\parens*{\sum_{\strategyalt\ni\resource}\flow_{\strategyalt}}.
\end{split}
\end{equation}

Since $\valueW(\rateprofalt)$ is finite for all $\rateprofalt\in \Rpos^{\commodities}$, it follows that $\vinf_{\rateprof}(\zvarprof)=\valueW(\rateprof+\zvarprof)$ is finite for $\zvarprof$ in  some interval around $\zerovec$,
and then the convex duality theorem implies that there is no duality gap and the subgradient $\nabla \vinf_{\rateprof}(\zerovec)$ at $\zvarprof=\zerovec$ coincides with the optimal solution 
set $\solset(\dual_{\rateprof})$ of the dual problem, that is, $\nabla \valueW(\rateprof)=\nabla \vinf_{\rateprof}(0)=\solset(\dual_{\rateprof})$. 

We claim that the dual problem has a unique solution, which is exactly the vector of equilibrium costs $\eqcost(\rate)$. 
Indeed,  fix an optimal solution $\eq{\flow}$ for $\vinf_{\rate}(0)
=\valueW(\rate)$ and recall that this is just a Wardrop equilibrium. 
The dual optimal solutions  are precisely the $\eqcost$'s in $\R^{\commodities}$ such that 
\begin{equation*}
\perturb_{\rate}(\eq{\flow},0)+\perturb_{\rate}^{*}(0,\eqcost)=0.
\end{equation*}
This equation can be written explicitly as
\begin{equation*}
\sum_{\resource\in\resources}\Cost_{\resource}\parens*{\sum_{\strategy\ni\resource}\eq{\flow}_{\strategy}}+\sup_{\flow\geq 0}\sum_{\commodity\in\commodities}\left(\eqcost^\commodity\parens*{\sum_{\strategy\in\strategies^{\commodity}} \flow_{\strategy}-\rate^\commodity}\right)-\sum_{\resource\in\resources}\Cost_{\resource}\parens*{\sum_{\strategyalt\ni\resource}\flow_{\strategyalt}}=0,
\end{equation*}
from which it follows that $\flow=\eq{\flow}$ is an optimal solution in the  latter supremum. The corresponding optimality conditions are
\begin{align*}
\eqcost^\commodity-\sum_{\resource\in\strategy}\cost_{\resource}\parens*{\sum_{\strategyalt\ni\resource}
\eq{\flow}_{\strategyalt}}=0,&\quad\text{ if }\eq{\flow}_{\strategy}>0,\commodity\in\commodities, \strategy\in\strategies^{\commodity},\\
\eqcost^\commodity-\sum_{\resource\in\strategy}\cost_{\resource}\parens*{\sum_{\strategyalt\ni\resource}
\eq{\flow}_{\strategyalt}}\leq 0,&\quad\text{ if }\eq{\flow}_{\strategy}=0,\commodity\in\commodities, \strategy\in\strategies^{\commodity},
\end{align*}
which imply that  $\eqcost^\commodity$ is the equilibrium cost of the \ac{OD} pair $\commodity$ for the Wardrop equilibrium, that is, $\eqcost^\commodity=\eqcost^{\commodity}(\rateprof)$ for every $\commodity\in\commodities$. 
It follows that the subgradient $\nabla \valueW(\rateprof)=\{\eqcost(\rateprof)\}$ so that $\rateprof\mapsto \valueW(\rateprof)$ is not only
convex but also differentiable with gradient $\nabla\valueW(\rateprof)=\eqcost(\rateprof)$. 
The conclusion follows by noting that every convex differentiable function is automatically of class $C^{1}$ and its gradient is monotone, in the sense that $\langle \nabla V(\rateprof_{1})- \nabla V(\rateprof_{2}),\rateprof_{1}-\rateprof_{2}\rangle\ge0$ for every $\rateprof_{1},\rateprof_{2}\in\Rpos^{\commodities}$, which in particular implies that $\eqcost^{\commodity}$ is nondecreasing in the variable $\rate^{\commodity}$.

The continuity of the equilibrium resource costs $\eqcostedge_{\resource}=\eqcostedge_{\resource}(\rateprof)$ is a consequence of Berge's maximum theorem \citep[see, \eg][Section~17.5]{AliBor:Springer2006}.
Indeed, as explained in \citet{Fuk:TRB1984}, the equilibrium resource costs are optimal solutions for the strictly convex dual program \eqref{eq:Fukushima}. 
Hence, since the objective function is jointly continuous in $(\eqcostedgeprof,\rateprof)$, Berge's theorem implies that the optimal solution correspondence is upper-semicontinous.
However, in this case the optimal solution is unique, so that the optimal correspondence is single-valued, and, as a consequence, the equilibrium resource costs $\eqcostedge_{\resource}(\rateprof)$ are continuous.
\end{proof}

\begin{remark}\label{rem:dual} A similar analysis where we reformulate the primal problem by including resource load variables $\load_\resource$ and considering perturbations in the flow balance equations $\load_{\resource}=\sum_{\strategy\ni\resource}\flow_{\strategy}+z_\resource$,
yields the dual problem \eqref{eq:Fukushima}, which characterizes the equilibrium costs $\eqcostedge_\resource$.

Perhaps a more direct argument is as follows.
Let us rewrite the flow balance equations 
$\eq\load_{\resource}=\sum_{\strategy\ni\resource}\eq\flow_{\strategy}$  in vector form as
$\eq\loadprof=\sum_{\strategy\in\strategies}\eq\flow_{\strategy}\canbresprof^{\strategy}$
where $\canbresprof^{\strategy}=(\canbres^{\strategy}_{\resource})_{\resource\in\resources}$ denotes the indicator vector 
with 
\begin{equation}
\label{eq:eta-r-s}
\canbres_{\resource}^{\strategy} = \mathds{1}_{\braces{\resource\in\strategy}}.   
\end{equation}
Since
$\rate_{\commodity}=\sum_{\strategy\in\strategies^{\commodity}}\eq\flow_{\strategy}$, by letting $\eq\alpha^{\commodity}_{\strategy}=\eq\flow_{\strategy}/\rate^{\commodity}$
for all $\strategy\in\strategies^{\commodity}$ we have that  $\sum_{\strategy\in\strategies^{\commodity}}\alpha^{\commodity}_{\strategy}=1$ and $\alpha^{\commodity}_{\strategy}\geq 0$, the latter inequality being strict only for the optimal strategies for commodity $\commodity$.
With these notations, we can write
\begin{equation}\label{eq:dualeq}
\eq\loadprof=\sum_{\strategy\in\strategies}\eq\flow_{\strategy}\canbresprof^{\strategy}=\sum_{\commodity\in\commodities}\rate_{\commodity}\sum_{\strategy\in\strategies^{\commodity}}\alpha^{\commodity}_{\strategy}\canbresprof^{\strategy}.
\end{equation}
Now, for each $\commodity\in\commodities$ the super-differential of the \emph{concave} function $\Theta_{\commodity}(\eqcostedgeprof)
\coloneqq \min_{\strategy\in\strategies^{\commodity}}\sum_{\resource\in\strategy}\eqcostedge_{\resource}$ is given by convex hull of the indicators of optimal strategies, that is,
\begin{equation}
\label{eq:partial-Theta}    
\partial\Theta_{\commodity}(\eqcostedgeprof)
=
\co\braces*{\canbresprof^{\strategy}: \strategy\in\strategies^{\commodity}, ~\sum_{\resource\in\strategy}\eqcostedge_{\strategy}=\Theta_{\commodity}(\eqcostedgeprof)},
\end{equation}
so that from \eqref{eq:dualeq} we derive
\begin{equation}\label{eq:dualeq2}
\eq\loadprof
=\sum_{\commodity\in\commodities}\rate_{\commodity}\sum_{\strategy\in\strategies^{\commodity}}\alpha^{\commodity}_{\strategy}\canbresprof^{\strategy}\in\sum_{\commodity\in\commodities}\rate_{\commodity}\,\partial\Theta(\eqcostedgeprof).
\end{equation}
Finally, letting $\Phi(\loadprof) \coloneqq \sum_{\resource\in\resources}\Cost_\resource(\load_\resource)    $ and
$\eqcostedge_{\resource}=\cost_{\resource}(\eq\load_{\resource})$ we clearly have  
$\eqcostedgeprof=\nabla\Phi(\eq\loadprof)$, which is equivalent to 
$\eq\loadprof\in\partial\Phi^*(\eqcostedgeprof)$ where the Fenchel's conjugate is given by $\Phi^*(\eqcostedgeprof)=\sum_{\resource\in\resources}\Cost^*_\resource(\eqcostedge_\resource)$. Since all the involved functions are finite and continuous, using
standard subdifferential calculus rules, \eqref{eq:dualeq2} is equivalent to $0\in\partial\Psi(\eqcostedgeprof)$ for the convex function
$\Psi(\eqcostedgeprof)=\Phi^*(\eqcostedgeprof)-\sum_{\commodity\in\commodities}\rate_{\commodity}\Theta_{\commodity}(\eqcostedgeprof)$ which is precisely the objective function in \eqref{eq:Fukushima}.
\end{remark}

\subsection{Characterization of comonotonicity}
For the sake of completeness we include the following  characterization of comonotonicity. This is a folk result (see \eg \citet{LanMei:AOR1994}), but its proof is not easy to find in the literature.

\begin{lemma}\label{lem:comonotonic}
Consider a finite family of functions $\comonf_i:\Omega\to\reals$ for $i=1,\ldots,m$ and let $s(\omega) \coloneqq \sum_{i=1}^{m} \comonf_i(\omega)$.
Then, the family $\{\comonf_i:i=1,\ldots,m\}$ is comonotonic if and only if there exist nondecreasing functions $F_i:\reals\to\reals$ such that
$\comonf_i(\omega)=F_i(s(\omega))$ for all $\omega\in\Omega$ and $i=1,\ldots,m$.
\end{lemma}

\begin{proof}
Since the \emph{``if'' } implication holds trivially, it suffices to prove the \emph{``only if''}. 
Suppose that the $\comonf_i$'s are comonotonic.
For $z\in\reals$ define
$F_i(z)=\sup_{\omega\in\Omega}\{\comonf_i(\omega):s(\omega)\leq z\}$ if there is some $\omega\in\Omega$ with $s(\omega)\leq z$,  and $F_i(z)=\inf_{\omega\in\Omega}\comonf_i(\omega)$ otherwise. 
Clearly the functions $F_i$ are  nondecreasing and $\comonf_i(\omega)\leq F_i(s(\omega))$, whereas comonotonicity implies that the latter holds with equality for all $i=1,\ldots,m$ and $\omega\in\Omega$.

It remains to show that the $F_i$'s are everywhere finite. Indeed, if $F_i(z)=\infty$ for some $z\in\reals$
we can find a sequence $\omega_n\in\Omega$  with $s(\omega_n)\leq z$ such that $\comonf_i(\omega_n)$ increases to $\infty$.
However, by comonotonicity, the latter implies $s(\omega_n)\to\infty$ which is a contradiction. 
Now, if $F_i(z)=-\infty$ we must be in the case $F_i(z)=\inf_{\omega\in\Omega}\comonf_i(\omega)=-\infty$ and comonotonicity implies  $\inf_{\omega\in\Omega} s(\omega)=-\infty$, so
we may find $\omega\in\Omega$ with $s(\omega)\leq z$ which yields the contradiction $F_i(z)\geq \comonf_i(\omega)>-\infty$.    
\end{proof}

\subsection{A remark on the number of active regimes}
\label{App:number_of_active_regimes}
The monotonicity result in \citet[proposition 3.12]{ComDosSca:MP2021} implies that the number of active regimes in a  single-commodity routing game on a \ac{SP} network is at most the number of paths. 
This bound does not hold for multiple commodities. 
In a singleton congestion game there are $\prod_{\commodity\in\commodities} \parens*{2^{\abs*{\resources^{\commodity}}}-1}$ potential combinations for $\activ{\resourcesprof}(\rateprof)$, and this bound may be attained (see \cref{ex:all-regimes-are-attained} below). 
This is not the case for single-commodity routing games: if we consider a subnetwork composed by only two paths, it is always \ac{SP} and only two of the three nonempty subsets of paths can actually correspond to an active regime $\activ{\resourcesprof}(\rate)$ for some $\rate\in[0,+\infty)$.

\begin{example}
\label{ex:all-regimes-are-attained}
Let us build a multi-commodity routing game that attains the maximal bound for the number of  active regimes.
Take $m$ a positive integer and consider a routing game on a parallel network with $m$ resources (edges)
$\resources=\{1,\ldots,m\}$ with cost functions
\begin{equation*}
\forall\irun\in \braces*{1,\dots,m},\quad\cost_\irun(\load_\irun)=\load_\irun+\irun,
\end{equation*}
and $m+1$ commodities where each commodity $\irun=1,\ldots,m$ can only use one specific resource $\resources^{\irun}=\{\irun\}$, whereas commodity $(m+1)$  can use all the resources $\resources^{m+1}=\resources$. 

We claim that $\activ{\resourcesprof}(\rateprof)$ assumes the maximum number 
$\prod_{\irun=1}^{m+1} \parens*{2^{\abs*{\resources^{\irun}}}-1} = 2^m-1$ of possible active regimes as the demands $\rateprof$ vary.
Indeed, for each commodity $\irun\leq m$ the active regime is always $\{\irun\}$, whereas every nonempty subset $\regime^{m+1}\subset\resources$
is the active regime of the
$(m+1)$-th commodity for some demand $\rateprof$. 
Namely, let $\irun_{\max}=\max\{\irun\in\regime^{m+1}\}$
and consider the demand 
\begin{align*}
\forall \irun\leq m, \quad\rate^{\irun}
&=
\begin{cases}
0 &\text{if }\irun\in\regime^{m+1},\\
\irun_{\max} &\text{if }\irun\notin\regime^{m+1},\\
\end{cases}\\
\rate^{m+1}
&=\sum_{\irun\in\regime^{m+1}}(\irun_{\max}-\irun).
\end{align*}
Then, the unique equilibrium is such that commodity $(m+1)$  allocates $\irun_{\max}-\irun$ to each resource $\irun\in\regime^{m+1}$ with cost $\irun_{\max}$, whereas every resource $\irun\notin\regime^{m+1}$ has a cost 
$\irun_{\max}+\irun >\irun_{\max}$, so that the active regime for commodity $(m+1)$ is exactly $\regime^{m+1}$.
\end{example}

%
%

\section{List of symbols}
\label{sc:list-of-symbols}

\begin{longtable}{p{.14\textwidth} p{.82\textwidth}}

$\cost_{\resource}$ & cost function of resource $\resource$\\
$\cost_{\strategy}$ & cost function of strategy $\strategy$, defined in \eqref{eq:strategy_cost}\\

$\costprof$ & vector of cost functions; it can be indexed both by elements in $\edges$ or $\routes$\\

$\cost^{\varepsilon}_\resource(\load_\resource)$ & regularized cost $ \cost_\resource(\load_\resource)+2\varepsilon \load_\resource$\\

$\Cost_\resource(\load_\resource)$ &  primitive of costs $ \int_0^{\load_\resource}\cost_\resource(z)\diff z$\\
$\Cost_{\resource}^{*}(\argdot)$ & Fenchel conjugate of $\Cost_\resource(\argdot)$\\
$\costclass$ & subset of commodities  having the same equilibrium cost\\

$\sink^{\commodity}$ & destination for \ac{OD} pair $\commodity$\\

$(\dual_{\rateprof})$ & dual problem, defined in \eqref{eq:dual}\\
$\edges$ &  set of edges\\

$\canb^{\commodity}$ &  $\commodity$-th vector of the canonical basis of $\mathbb{R}^{\commodities}$\\

$\flow_{\strategy}^{\commodity}$ & flow on strategy $\strategy$ in commodity $\commodity$\\

$\flowprof^{\commodity}$ & $\commodity$-th commodity flow vector $ \parens*{\flow_{\strategy}^{\commodity}}_{\strategy\in\strategies^{\commodity}}$\\
$\flowprof$ &flow vector $\parens*{\flowprof^{\commodity}}_{\commodity\in\commodities}$ \\
$\flows_{\rateprof}$ & set of feasible pairs $(\loadprof,\flowprof)$ for the demand vector $\rateprof$\\
$\graph$ &  directed multigraph\\
$\game$ & $\parens*{\resources,\costprof,\strategiesprof}$, congestion game structure\\

$\game^{\varepsilon}$ & $(\resources,\costprof^{\varepsilon},\strategiesprof)$,  perturbed congestion game structure\\
$\game_\costclass$ & $(\resources_\costclass,\costprof,\strategies_{\costclass} )$, single commodity game defined in \cref{def:region-order}\\
$\game_{1} \times \game_{2}$ & series game\\
$\game_{1} \cup \game_{2}$ & parallel game\\
$\commodity$ & commodity\\
$\commodities$ & set of commodities\\
$\source^{\commodity}$ & origin for \ac{OD} pair $\commodity$\\
$(\prim_{\rateprof})$ & primal problem, defined in \eqref{eq:primal}\\
$\activ{\routes}(\rate)$ & set of paths that attain equilibrium cost at equilibrium with demand $\rate$ \\
$\routes^{\commodity}$ &  the set of paths of commodity $\commodity$\\

$\routesprof$ & $(\routes^{\commodity})_{\commodity\in\commodities}$\\

$\resource$ & resource\\
$\resources$ & set of resources\\
$\activ{\resources}^{\commodity}(\rateprof)$ & set of active resources for commodity $\commodity\in\commodities$\\
$\activ{\resourcesprof}(\rateprof)$ & $(\activ{\resources}^{\commodity}(\rateprof))_{\commodity\in\commodities}$, active regime\\

$\resources_0$ &  set of resources such that $\cost_{\resource}(\load_{\resource}(\rateprof_0))=\eqcost^{\commodity}(\rateprof_0)$\\

$\resources_0^{+}$ &
$\{\resource\in \resources_0\colon \load_\resource(\rateprof_0+\var \canb^{\commodity})>\load_\resource(\rateprof_0)\}$, defined in \eqref{eq:R0+}\\
$\resources_0^{-}$ & $\{\resource\in \resources_0\colon \load_\resource(\rateprof_0+\var \canb^{\commodity})<\load_\resource(\rateprof_0)\}$, defined in \eqref{eq:R0-}\\
$\resources_0^{=} $ & 
$\{\resource\in \resources_0\colon \load_\resource(\rateprof_0+\var \canb^{\commodity})=\load_\resource(\rateprof_0)\}$, defined in \eqref{eq:R0=}\\
$\resources_{\costclass}$ & $\parens*{\cup_{\commodity\in\costclass}\resources^{\commodity}} \setminus \parens*{\cup_{\commodityalt\succ\costclass}\resources^{\commodityalt}}$, defined in \eqref{eq:R-C}\\

$\strategies^{\commodity}$ & set of feasible strategies for commodity $\commodity$\\
$\strategiesprof$ & $ \times_{\commodity\in\commodities}{\strategies^{\commodity}}$, set of strategy profiles\\
$\solset(\dual_{\rateprof})$ &  optimal solution 
set of the dual problem\\
${\SC}(\rateprof)$ & $\sum_{\commodity\in \commodities} \rate^{\commodity}\eqcost^{\commodity}(\rateprof)$,  social cost\\

$\vinf_{\rateprof}(\zvarprof)$ & $\inf_{\flowprof}\perturb_{\rateprof}(\flowprof,\zvarprof)$, defined in \eqref{eq:primal}\\
$\valueW(\rateprof)$ & $\min_{\loadprof\in\loads_{\rateprof}}\sum_{\resource\in\resources}\Cost_\resource(\load_\resource)$, defined in \eqref{eq:min-Beckmann}\\
$\vertices$ & set of vertices\\

$\load_{\resource}$ & load of resource $\resource$, defined in \eqref{eq:edge-loads}\\
$\loadprof$ & $ (\load_{\resource})_{\resource\in\resources}$, load vector\\ 
$\loads_{\rateprof}$  &
projection of the set of feasible pairs $\flows_{\rateprof}$ onto the $\loadprof$ variables\\

$\region^\precsim$ & demand regions induced by a given order $\precsim$, defined in \eqref{eq:Gamma-prec}\\
$\region^{\precsim}_{\regimeprof}
$ & demand subregions for a given order $\precsim$ and active regime $\regime$, defined in \eqref{eq:Gamma-rho-prec} \\

$\canbres_{\resource}^{\strategy}$ & $\mathds{1}_{\braces{\resource\in\strategy}}$, defined in \eqref{eq:eta-r-s}\\
$\canbresprof^{\strategy}$ & $(\canbres^{\strategy}_{\resource})_{\resource\in\resources}$\\

$\eqcost^{\commodity}$ & equilibrium cost of commodity $\commodity$, defined in \eqref{eq:Wardrop}\\
$\eqcostprof$ & $\parens*{\eqcost^{\commodity}}_{\commodity\in\commodities}$, equilibrium cost vector\\

$\rate^{\commodity}$ & demand for commodity $\commodity$\\
$\rateprof$ & $\parens*{\rate^{\commodity}}_{\commodity\in\commodities}$ demand vector\\

$\rate_{\costclass}$ & aggregate demand on $\costclass$\\

$\regimeprof$ & $\parens*{\regime^{\commodity}}_{\commodity\in\commodities}$, regime\\
$\regime^{\commodity}$ & subset of $\routes^{\commodity}$\\
$\eqcostedge_{\resource}$& equilibrium cost of resource $\resource$\\

$\perturb_{\rateprof}$ & defined in \eqref{eq:perturb}\\
$\perturb_{\rateprof}^{*}$ & Fenchel conjugate of $\perturb_{\rateprof}$, defined in \eqref{eq:phi-star}\\

\end{longtable}

\bibliographystyle{apalike}
\bibliography{biblio-games}

\end{document}